\documentclass[review]{elsarticle}

\usepackage{lineno,hyperref}

\usepackage{framed,multirow}

\usepackage{times}
\usepackage{epsfig}
\usepackage{graphicx}
\usepackage{pifont}
\usepackage{graphics}
\usepackage{printlen}
\usepackage{booktabs,subcaption,amsfonts,dcolumn}
\usepackage[dvipsnames]{xcolor}
\usepackage{subcaption}
\usepackage{amsmath,amssymb,amsfonts}
\usepackage{graphicx}
\usepackage{textcomp}
\usepackage{mathtools}
\usepackage{amsthm}
\usepackage{algorithm}
\usepackage{algpseudocode}
\usepackage{multirow}
\usepackage{setspace}
\theoremstyle{definition}

\let\Algorithm\algorithm
\renewcommand\algorithm[1][]{\Algorithm[#1]\setstretch{1.2}}
\usepackage{eqparbox}
\usepackage{array}
\algnewcommand{\LeftComment}[1]{\Statex \(\triangleright\) #1}

\newtheorem{theorem}{Theorem}
\newtheorem{definition}{Definition}[theorem]
\newtheorem{lemma}[theorem]{Lemma}
\newtheorem{proposition}[theorem]{Proposition}
\usepackage{booktabs,subcaption,amsfonts,dcolumn}
\newcolumntype{d}[1]{D..{#1}}
\newcommand{\splitcell}[1]{\begin{tabular}{@{}c@{}}#1\end{tabular}}

\usepackage{amssymb}
\usepackage{latexsym}

\usepackage{url}
\usepackage{xcolor}

\usepackage{hyperref}
\usepackage{amsmath,amssymb,amsfonts}
\usepackage{graphicx}
\usepackage{textcomp}
\usepackage{multirow}
\usepackage{caption}
\usepackage{url}
\usepackage{hyperref}
\usepackage{soul}
\usepackage{cleveref}
\usepackage{makecell}
\usepackage{xcolor}
\usepackage{tikz}
\usepackage{float}

\usepackage{latexsym}
\usepackage{rotating}
\usepackage{enumitem} 
\modulolinenumbers[5]
\journal{Journal of Medical Image Analysis}

\begin{document}

\begin{frontmatter}

\title{Multi-rater Prism: Learning self-calibrated medical image segmentation from multiple raters}

\author[2]{Junde Wu}
\author[2]{Huihui Fang}
\author[2]{Yehui Yang}
\author[3]{Yuanpei Liu}
\author[4]{Jing Gao}
\author[5]{Lixin Duan}
\author[6]{Weihua Yang}
\author[2]{Yanwu Xu\corref{cor1}}  

\address[2]{Baidu Inc., Beijing, China}
\address[3]{Department of Statistics and Actuarial Science, The University of Hong Kong, HongKong, China}
\address[4]{Electrical and Computer Engineering, Purdue University, West Lafayette, Indiana, USA}
\address[5]{School of Computer Science and Technology, University of Electronic Science and Technology of China, Chengdu, Sichuan, China}
\address[6]{Big Data and Artificial Intelligence Institute, Shenzhen Eye Hospital, Jinan University, Shenzhen, Guangdong, China}
\cortext[cor1]{Corresponding authors: Yanwu Xu (ywxu@ieee.org).}



\begin{abstract}
In medical image segmentation, it is often necessary to collect opinions from multiple experts to make the final decision. This clinical routine helps to mitigate individual bias. But when data is multiply annotated, standard deep learning models are often not applicable. In this paper, we propose a novel neural network framework, called Multi-Rater Prism (MrPrism) to learn the medical image segmentation from multiple labels. Inspired by the iterative half-quadratic optimization, the proposed MrPrism will combine the multi-rater confidences assignment task and calibrated segmentation task in a recurrent manner. In this recurrent process, MrPrism can learn inter-observer variability taking into account the image semantic properties, and finally converges to a self-calibrated segmentation result reflecting the inter-observer agreement. Specifically, we propose Converging Prism  (ConP) and Diverging Prism  (DivP) to process the two tasks iteratively. ConP learns calibrated segmentation based on the multi-rater confidence maps estimated by DivP. DivP generates multi-rater confidence maps based on the segmentation masks estimated by ConP. The experimental results show that by recurrently running ConP and DivP, the two tasks can achieve mutual improvement. The final converged segmentation result of MrPrism outperforms state-of-the-art  (SOTA) strategies on a wide range of medical image segmentation tasks.
\end{abstract}
\end{frontmatter}


\begin{figure}[!h]
\centering
\includegraphics[width=0.8\textwidth]{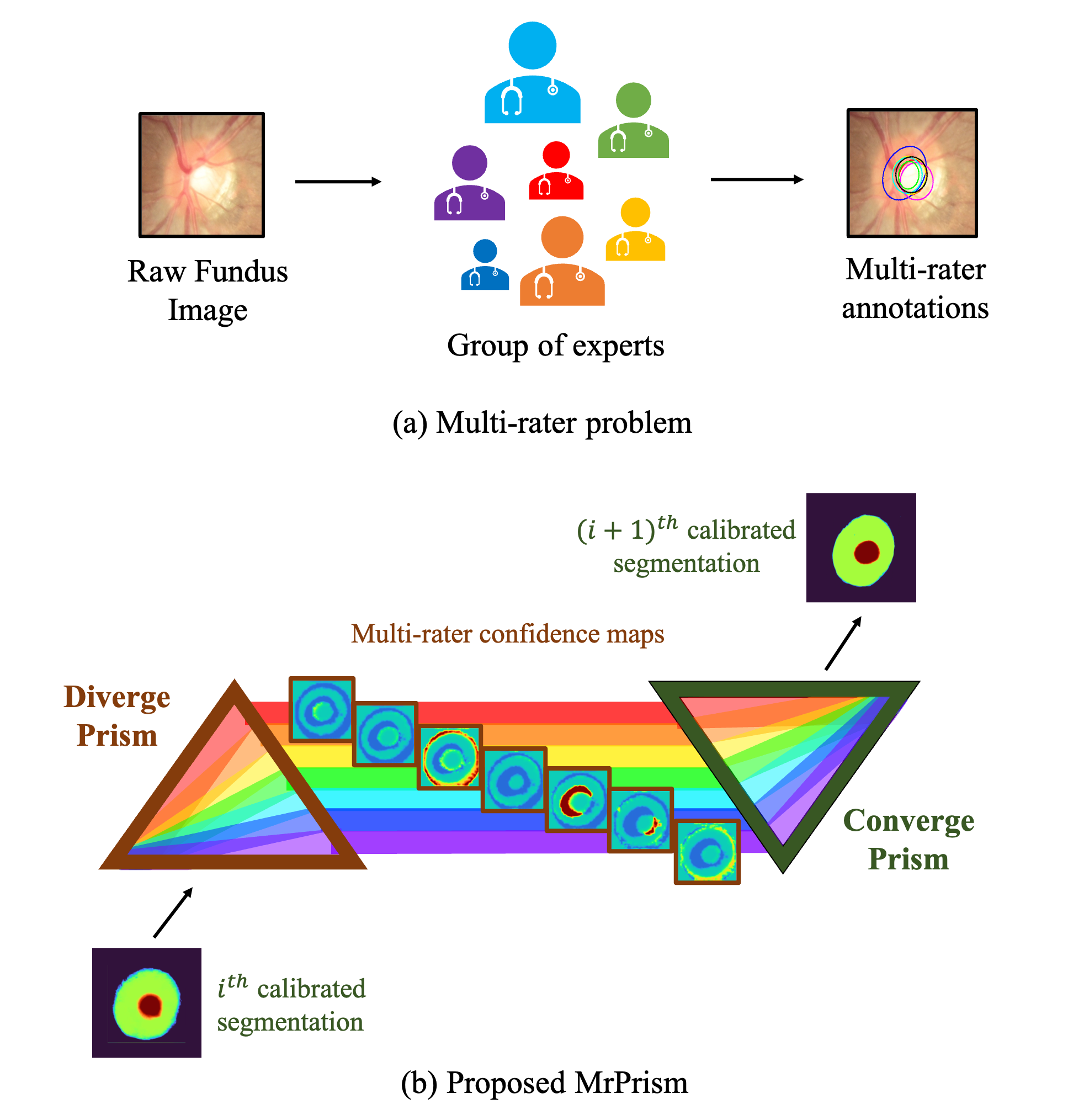}
\caption{ (a). A multi-rater example of optic-cup annotation. We can see that the inner-variance among the annotations is large.  (b). In the proposed MrPrism solution, Diverge Prism estimates the multi-rater confidences from the calibrated segmentation mask provided by Converge Prism. Converge Prism again predicts the calibrated segmentation mask from the multi-rater confidences provided by Diverge Prism. Through such an iterative optimization, the segmentation results can be gradually calibrated and refined.}
\vspace{-10pt}
\label{fig:prisms}
\end{figure}

\section{Introduction}\label{sec:introduction}
Clinically, it is often necessary to integrate the opinions of several different clinical experts to make the final decision. Thus, in the medical image analysis tasks, it is common to collect the data with multiple labels annotated by different clinical experts. Especially in medical image segmentation, the labels are likely to be highly subjective. An example of optic cup/disc segmentation on fundus images is shown in Fig. \ref{fig:prisms}  (a). We can see that the variance between different annotations is large. It makes the deep learning models which work well on nature images can not be directly applied to such a scenario. This problem is called 'multi-rater problem' by the prior works (\cite{warrens2010inequalities,ji2021learning,wu2022opinions}).

The multi-rater problem is previously studied from two different perspectives. The first branch of study learns calibrated results to reflect the inter-observer variability  (\cite{guan2018said,chou2019every,jensen2019improving,ji2021learning}). A body of techniques are proposed, including Monte Carlo dropout (\cite{kendall2015bayesian,kendall2017uncertainties}), ensembles (\cite{lee2016stochastic, lakshminarayanan2017simple}), multi-head networks (\cite{guan2018said, rupprecht2017learning}), and variational Bayesian inference (\cite{baumgartner2019phiseg,kohl2018probabilistic,kohl2019hierarchical}). They often adapted the model to be supervised by multiple labels instead of certain unique fusion. In this way, the model learns to represent the underlying  (dis-)agreement among multiple experts. Recently, Ji \MakeLowercase{\textit{et al.}} (\cite{ji2021learning}) shows that such calibrated results can achieve better performance on different kinds of fused labels. However, these methods often assume that each rater is equally important.
Thus, they are still limited to predicting the traditional majority vote  (average of multiple labels).

Another branch of work is to find the potential correct label from the multi-rater labels (\cite{warfield2004simultaneous,raykar2010learning,wu2022opinions}). Commonly, the confidence of each rater will be evaluated to calculate a weight fusion of multi-rater labels as the final unique label. One general and reliable way is to evaluate the multi-rater confidences based on the raw image structural prior (\cite{rodrigues2018deep,cao2019max,tanno2019learning,albarqouni2016aggnet}). Under such a prior constraint, annotations with a structure more similar to that observed on the raw image can obtain higher confidence.
However, these methods learned fused label without calibration. Therefore, the multi-rater confidences cannot be dynamically adjusted in the inference stage, which causes the results to be either overconfident  (when the confidences are implicitly learned (\cite{rodrigues2018deep,tanno2019learning})) or ambiguous  (when confidences are explicitly learned (\cite{cao2019max,albarqouni2016aggnet})). 



We note that on the one hand, the calibrated segmentation needs to know the confidence of the raters to predict the segmentation mask. On the other hand, multi-rater confidences assignment lacks a calibrated segmentation model to improve the dynamic representation capability.
In order to make up for the shortcomings of these two branches of study, we propose a novel framework, named Multi-Rater Prisms  (MrPrism), to jointly learn the calibrated segmentation and multi-rater confidences assignment in a recurrent manner. In this way, the two tasks can be jointly optimized to achieve mutual improvement. After the proposed recurrent learning from the multi-rater annotations, MrPrism can arrive at a self-calibrated segmentation mask in favor of the inter-observer agreement.

The proposed recurrence process follows the guidance of the iterative half-quadratic optimization, which simultaneously optimizes the calibrated segmentation and the multi-rater confidences under the constraint of the raw image structural prior. 
Specifically, proposed MrPrism consists of Converging Prism  (ConP) and Diverging Prism  (DivP). ConP learns calibrated segmentation based on the multi-rater confidences provided by DivP. It takes raw image segmentation features as input and multi-rater confidences as conditions to predict a unique segmentation mask. In this way, the raw image structural prior can be implicitly learned in the neural network. DivP learns to estimate multi-rater labels from the segmentation masks provided by ConP. The estimated probability maps of the raters can be used to represent the confidences of them naturally  (proved in Proposition 1).
Through the iterative optimization of DivP and ConP, both calibrated segmentation and multi-rater confidences can gradually converge to the optimal solutions under raw image structural prior.

ConP and DivP are mainly implemented by vision transformers to capture the tissue/organ structure globally. In particular, ConP is implemented by a combination of dot-product attentions and deconvolution layers. Thanks to the global and dynamic nature of the attentive mechanism, ConP can naturally integrate the multi-rater confidences dynamic representations to the segmentation features. DivP is implemented by the multi-head dot-product attentions. Each head of DivP estimates the corresponding rater annotation, and the estimated possibility maps are used to represent the multi-rater confidences. 
To avoid the trivial solution, we also shuffle the multi-head loss function of DivP to disentangle the multi-rater segmentation maps and confidences. 
The experimental results show that the final segmentation outperforms previous SOTA methods on various medical image segmentation tasks.


In brief, the paper's four main contributions are as follows:

1. We propose MrPrism, a novel medical image segmentation framework to learn from multi-rater annotated labels. Guided by the iterative half-quadratic optimization, MrPrism can calibrate the segmentation by itself through the recurrent learning. 


2. In MrPrism, we design ConP and DivP to learn the calibrated segmentation and multi-rater confidences assignment in a recurrent way, which helps to gain the mutual improvement on both tasks. 

3. We carefully design the structure and supervision of MrPrism for the iterative optimization. The techniques proposed, including attentive feature integration and shuffled supervision, are found effective for the recurrent learning.

4. We conduct comprehensive experiments on a wide range of medical segmentation tasks. The experimental results show that MrPrism consistently outperforms the previous SOTA strategies. The improvement is especially significant when the inter-rater variety is large, reflecting the superior ability of MrPrism to tackle the multi-rater disagreement.

\section{Related Works}\label{sec:ralated works}
Learning from the multi-rater labels/labels from crowds is a historical study that can be retrospected to the work of Dawid and Skene (\cite{dawid1979maximum}). Then a long list of papers was based on their method to forecast the accurate label from crowds. Literature (\cite{ ghosh2011moderates,dalvi2013aggregating,karger2014budget}) adopted the generative model of  (\cite{dawid1979maximum}) and extended it to multiple scenarios under certain assumptions. The expectation-maximization (EM) based methods were popular (\cite{raykar2010learning,rodrigues2018deep,albarqouni2016aggnet,cao2019max}) in the last decade. The basic idea is to estimate the posterior probability of each rater by taking the classifier predictions as the prior. Rodrigues \MakeLowercase{\textit{et al.}} (\cite{rodrigues2018deep}) extended this idea to neural-network-based classifiers to facilitate the learning. However, these methods generally adopted standard classifiers without calibration, making them inadaptable in the inference stage.

Another branch of study focused on capturing the inter-rater variety by learning from the multi-rater labels (\cite{guan2018said,chou2019every,ji2021learning,rupprecht2017learning,jensen2019improving}). One way is to adopt a label sampling strategy to sample labels randomly from the multi-rater labeling pool during training (\cite{jensen2019improving,jungo2018effect}). The results were observed to be better calibrated than training on the traditional majority vote. Another way is to model each rater individually by multiple decoders in a neural network (\cite{guan2018said,chou2019every}). They showed that the multi-head supervision outperforms that of a unique combined ground-truth. Ji \MakeLowercase{\textit{et al.}} (\cite{ji2021learning}) proposed a calibrated model to estimate the corresponding segmentation masks based on any given multi-rater confidence. The model obtained remarkable performance on various fused labels. However, it still needs users to provide the correct multi-rater confidence in advance, which limits its application.

\section{Theoretical Analysis }\label{sec:assumptions}
In this section, we formally define the problem and give the theoretical premises of MrPrism.

Suppose that there are $M$ raters, and $K$ classes, e.g. {optic disc, optic cup, background} in optic cup/disc segmentation task. Denote by a matrix $z^{m} \in \mathbb{R}^{H\times W \times K}$ the observed label that rater $m$ annotates an item $x$, $H$ and $W$ are the height and width of the item respectively. Let $z^{[M]}$ denotes ${z^{1},z^{2},...,z^{M}}$. The data points $x$ and multi-rater labels $z^{1}, z^{2}...z^{M}$ are assumed to be drawn i.i.d. from random variables $X$ and $Z^{1}, Z^{2}...Z^{M}$. 

Denote by $y$ that the fusion of $z^{[M]}$ by multi-rater confidence maps $w^{[M]}$, which can be expressed as:
\begin{equation}\label{equ:gt} 
    y = softmax (\sum_{m = 1}^{M} w^{m} \cdot z^{m} + p)
\end{equation}
where $\cdot$ represents element-wise multiplication of two matrix, $p$ is the prior probability of $y$ segmented from data point $x$. We softmax the matrix dimension which represents the classes, to make sure the sum of the possibility is one. 

Then, we can formally model the issue as:
\begin{equation}\label{equation:intact}
 \arg min_{\mathcal{W}, \mathcal{Y}} \Vert \mathcal{W} \otimes \mathcal{Z} - \mathcal{Y}\Vert^{2}_{2}+ \eta \space \mathcal{P}_{x},
\end{equation}
where $\mathcal{W}$ denotes the optimization variables w.r.t the multi-rater confidences, $\mathcal{Y}$ denotes the calibrated segmentation ground-truth which can be fully represented by the optimal confidences and prior according to Eqn (\ref{equ:gt}), $\mathcal{Z}$ denotes the observed multi-rater label matrix. \(\otimes\) implies general convolution operation, which we define $\mathcal{W} \otimes \mathcal{Z} = softmax (\sum_{m = 1}^{M} w^{m} \cdot z^{m} + p_{u})$, where $p_{u}$ is the uniform prior (a matrix with constant values), \(\mathcal{P}_{x}\) is the constraint image prior related to the raw image $x$, \(\eta\) is a weighting constant. Our goal is to estimate the confidences $\mathcal{W}$ together with the calibrated mask $\mathcal{Y}$ which minimize the equation given multi-rater labels $\mathcal{Z}$ and raw image $x$.

Directly solving Eqn (\ref{equation:intact}) is hard since both two terms contain the unknown optimization variables. We notice that the problem can be simplified by the iterative half-quadratic optimization (\cite{geman1992constrained}.  By half-quadratic minimization, Eqn.  (\ref{equation:intact}) is equivalent to the iterative optimization of following equations:
\begin{equation}\label{equation:HQ}
\left\{\begin{array}{ll}
\begin{split}
\  \mathcal{W}^{'}_{i}= &\arg min_{\mathcal{W}^{'}_{i}} \frac{\beta}{2}\Vert \mathcal{W}_{i-1} - \mathcal{W}^{'}_{i}\Vert^{2}_{2} 
\\&+ \eta \; \mathcal{P}_{x}(\mathcal{W}^{'}_{i}) & {\raisebox{.5pt}{\textcircled{\raisebox{-.9pt} {1}}}}
\\
\  \mathcal{W}_{i} = &\arg min_{\mathcal{W}_{i}} \frac{\beta}{2} \Vert \mathcal{W}_{i} - \mathcal{W}^{'}_{i}\Vert^{2}_{2} \\&+\frac{1}{2}\Vert \mathcal{W}_{i} \otimes Z - \mathcal{Y}_{i}\Vert^{2}_{2}& {\raisebox{.5pt}{\textcircled{\raisebox{-.9pt} {2}}}}
\end{split}
\\
\end{array}
\right.
\end{equation}
where $\mathcal{W}^{'}_{i}$ is an auxiliary variable introduced to relate two equations, \(i\) is the number of iterations. \(\beta\) is a variable parameter. Eqn.  (\ref{equation:HQ}) can be solved by alternatively solving two sub-problems with increasing \(\beta\). 

Such an observation inspired us to design two sub-models to fit the two sub-problems respectively and run iteratively to solve the final optimization problem. In particular, we design ConP to fit the sub-problem (\ref{equation:HQ})-\(\raisebox{.5pt}{\textcircled{\raisebox{-.9pt} {1}}}\)  and DivP to fit the sub-problem (\ref{equation:HQ})-\(\raisebox{.5pt}{\textcircled{\raisebox{-.9pt} {1}}}\). To conceptually understand the design, we can see the first terms of Eqn.  (\ref{equation:HQ})-\(\raisebox{.5pt}{\textcircled{\raisebox{-.9pt} {1}}}\) and (\ref{equation:HQ})-\(\raisebox{.5pt}{\textcircled{\raisebox{-.9pt} {2}}}\) optimize their current optimizing variables to get enough close to the other's last predicted result. We fulfill this condition in MrPrism by supervising the two sub-models with each other's last predicted result in the training stage. The second term of Eqn.  (\ref{equation:HQ})-\(\raisebox{.5pt}{\textcircled{\raisebox{-.9pt} {1}}}\) is a constraint of image prior. We make it to be self-learned implicitly by inputting the raw images to the neural network based sub-model (ConP). Because of the continuity nature of neural network (similar inputs tend to get similar outputs) (\cite{Zhang2017Learning,dong2018denoising,wu2020integrating}), the network is prone to output the segmentation mask following the raw structure of the image. 


The second term of Eqn.  (\ref{equation:HQ})-\(\raisebox{.5pt}{\textcircled{\raisebox{-.9pt} {2}}}\) encourages the optimal multi-rater confidences $\mathcal{W}_{i}$ given the multi-rater labels $Z$ and correctly fused mask $\mathcal{Y}_{i}$. However, we know neither the optimal multi-rater confidences and the correctly fused mask, which puts it hard for the supervision. To address the issue, we find a relationship between them that allows us to directly supervise the sub-model through the multi-rater labels, which can be represented as: 
\begin{proposition}\label{pro:content1}
The confidence map ${w}^{m}$ can be obtained from the natural logarithm of the estimated probability map of the label $z^{m}$ given the fused mask. That is, ${w}^{m} = {\rm log} \;P (z^{m} | y)$.
\end{proposition}
The proof of Proposition \ref{pro:content1} is put in the supplement materials.

Based on Proposition \ref{pro:content1}, we can supervise DivP though the multi-rater labels $z^{M}$ and take its output as the multi-rater confidence maps. In this way, the constraint in Eqn.  (\ref{equation:HQ})-\(\raisebox{.5pt}{\textcircled{\raisebox{-.9pt} {2}}}\) can be constructed by the supervision of the fused masks, as the shuffle loss $\mathcal{L}_{sff}$ we designed in Section \ref{sec:shuffle}. To put it more clearly, we first 
define two kinds of fused masks according to Proposition \ref{pro:content1}, which we use widely in the later discussion. 
\begin{definition}
\begin{equation}\label{equ:gt2}
    y^{fuse} = \tilde{z}^{[M]} \odot z^{[M]} = softmax (\sum_{m = 1}^{M} {\rm log} \; (\tilde{z}^{m}) \cdot z^{m} + p_{u}),
\end{equation}
\begin{equation}\label{equ:gt2}
    y^{self} = \tilde{z}^{[M]} \odot z^{[M]} = softmax (\sum_{m = 1}^{M} {\rm log} \; (\tilde{z}^{m}) \cdot \mathcal{T}  (\tilde{z}^{[M]}) + p_{u}),
\end{equation}
\end{definition}
where $\mathcal{T}$ denotes the thresholding, $\odot$ denotes the fusion operation. Then we can transfer Eqn. (\ref{equation:HQ}) to a format more convenient for the implementation:
\begin{proposition}\label{pro:2}
Solving Eqn. (\ref{equation:HQ}) is equivalent to solve the following equation:
\begin{equation}\label{equation:HQ-adp}
\left\{\begin{array}{ll}
\begin{split}
\  \mathcal{Y}^{'}_{i}= &\arg min_{\mathcal{Y}^{'}_{i}} \frac{\beta}{2}\Vert \mathcal{Y}^{self}_{i-1} - \mathcal{Y}^{'}_{i}\Vert^{2}_{2} 
\\&+ \eta \; \mathcal{P}_{x}(\mathcal{Y}^{'}_{i}) & {\raisebox{.5pt}{\textcircled{\raisebox{-.9pt} {1}}}}
\\
\  \mathcal{Y}^{self}_{i} = &\arg min_{\mathcal{Y}^{self}_{i}} \frac{\beta + 1}{2} \Vert \mathcal{Y}^{self}_{i} - \mathcal{Y}^{'}_{i}\Vert^{2}_{2} \\&+\frac{1}{2}\Vert \mathcal{Y}^{fuse}_{i} - \mathcal{Y}^{self}_{i}\Vert^{2}_{2}& {\raisebox{.5pt}{\textcircled{\raisebox{-.9pt} {2}}}}
\end{split}
\\
\end{array}
\right.
\end{equation}
\end{proposition}
Likewise, $\mathcal{Y}^{'}$ is an auxiliary variable to relate two sub-equations. The proof of Proposition \ref{pro:2} is provided in the supplement materials.

In Eqn. (\ref{equation:HQ-adp}), we reorganize all the variables to the same form (multi-masks fusions), which can largely facilitate our practical implementation. It not only allows us to supervise the two sub-models in a more consistent way but also allows ConP to only produce the fusions with the confidences and image prior implicitly learned. 

In the practice, we train an iteratively optimized framework consists of ConP and DivP, called MrPrism. In MrPrism, ConP is designed as a segmentation decoder that predicts the fusion $\mathcal{Y}^{'}$ based on the raw image $x$. The network is supervised by $\mathcal{Y}^{self}$ produced by DivP in the last iteration (the first term in (\ref{equation:HQ})-\(\raisebox{.5pt}{\textcircled{\raisebox{-.9pt} {1}}}\)). The raw image prior is implicitly learned from the inputted raw image (the second term in (\ref{equation:HQ})-\(\raisebox{.5pt}{\textcircled{\raisebox{-.9pt} {1}}}\)). DivP is designed as a multi-rater masks predictor that estimates the multi-rater confidences $\mathcal{W}$ based on the segmentation masks predicted by ConP $\mathcal{Y}^{'}$ in the last iteration. Instead of the direct supervision on the confidences, we supervise the self-fusion of the confidences, i.e., $\mathcal{Y}^{self}$ with the last fusion estimated by ConP $\mathcal{Y}^{'}$ (the first term of (\ref{equation:HQ})-\(\raisebox{.5pt}{\textcircled{\raisebox{-.9pt} {2}}}\)) and the fusion with multi-rater labels, i.e., $\mathcal{Y}^{fuse}$ (the second term of (\ref{equation:HQ})-\(\raisebox{.5pt}{\textcircled{\raisebox{-.9pt} {2}}}\)). Such an iterative optimization process will converge to the consistent agreement between ConP and DivP constraint by the raw image structural prior according to the half-quadratic algorithm.

\section{Methodology}
\begin{figure*}[!t]
\centering
\includegraphics[width=\textwidth]{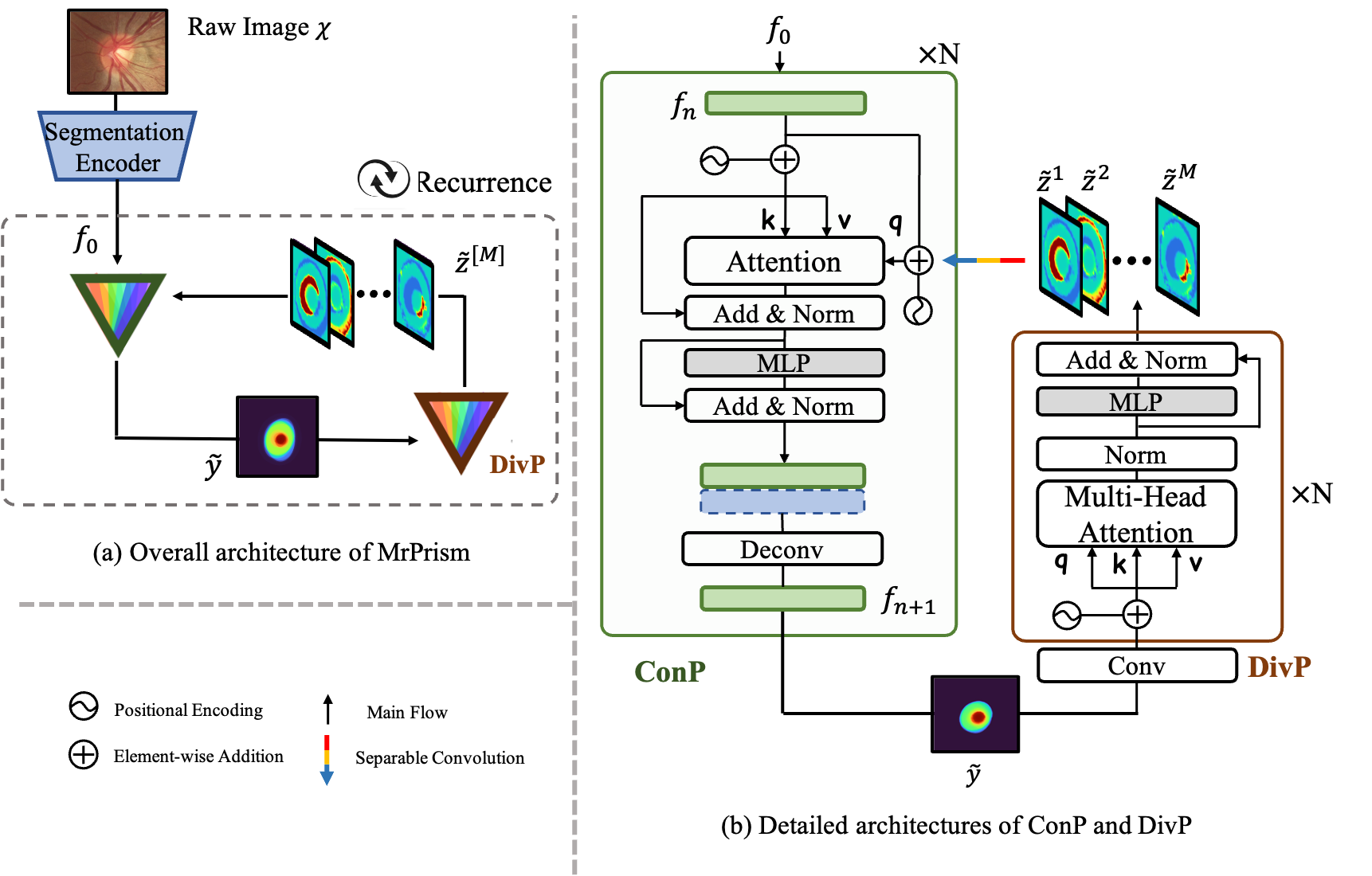}
\caption{Overall architecture of MrPrism. Green denotes ConP modules. Orange denotes DivP modules.}
\label{fig:overall}
\end{figure*}
The overall flow of the proposed method is shown in Fig. \ref{fig:overall}. Raw image $x$ is first sent into a CNN-based encoder to obtain a deep embedding $f$. Then ConP utilizes $f$ and given confidences $\tilde{w}^{[M]}$ to estimate the calibrated segmentation mask $\tilde{y}$. 
DivP will then separate $\tilde{y}$ to multi-rater segmentation masks. Each head of DivP estimates the corresponding rater annotation. The estimated probability maps can represent the multi-rater confidence maps as shown in Proposition \ref{pro:1}. Multi-rater confidence maps will then be embedded and sent to ConP for the calibration in the next iteration. ConP and DivP will run recurrently until converge. 
Since the paper cannot cover all the details of the network design, we introduce the implementation of ConP and DivP modules from the methodology level in the following paragraphs. The detailed network structures can be found in our released code.

\subsection{ConP}
We propose ConP to estimate calibrated segmentation masks based on estimated multi-rater confidences. The basic structure is shown in Fig. \ref{fig:overall}  (b). The input of ConP is the raw image embedding $f_{0}$, and the output is the segmentation mask $\tilde{y}$. Multi-rater confidence maps $\tilde{w}^{[M]}$ estimated by DivP are integrated into ConP through attentive mechanism to calibrate the segmentation. In ConP, attention is inserted into each two of the deconvolution layers. It takes embedded multi-rater confidence as \textit{query}, segmentation features as \textit{key} and \textit{value}. In this way, the segmentation features can be selected and enhanced based on the given multi-rate confidence maps. 

Specifically, consider ConP at the $n^{th}$ layer, the segmentation feature is $f_{n} \in \mathbb{R}^{\frac{H}{r_{n}} \times \frac{W}{r_{n}} \times C_{n}}$, where $ (\frac{H}{r_{n}},\frac{W}{r_{n}})$ is the resolution of the feature map, $C_{n}$ is the number of channels. The embedding of multi-rater confidence maps is $\tilde{w}^{e}_{n} \in \mathbb{R}^{\frac{H}{r_{n}} \times \frac{W}{r_{n}} \times C_{n}}$.Then $f_{n}$ is transferred by:
\begin{equation}\label{equation:attconv}
     \bar{f}_{n} = {\rm Attention} (q, k, v) = {\rm Attention} (\tilde{w}^{e}_{n} + E_{w},f_{n} + E_{f}, f_{n}).
\end{equation}
where ${\rm Attention} (query, key, value)$ denotes attention mechanism, $E_{w}, E_{f}$ are positional encodings (\cite{carion2020end} for confidence embedding and segmentation feature map respectively. Following  (\cite{dosovitskiy2020image}, we reshape the feature maps into a sequence of flattened patches before the attention. Similarly, $\bar{f}_{n}$ will be reshaped back to $\mathbb{R}^{\frac{H}{r_{n}} \times \frac{W}{r_{n}} \times C_{n}}$ after the attention. 



The results will then pass through a multi-layer perceptron  (MLP) layer to further reinforce the targeted features. The residual connection (\cite{resnet}, followed by layer normalization (\cite{ba2016layer} are employed before and after MLP layer to facilitate the training. MLP layer consists of two linear mappings that keep the dimension of the input. Then a standard UNet (\cite{ronneberger2015u} operation is adopted to up-sample the feature. It will be first concatenating with the corresponding encoder feature and then applying deconvolution layer to obtain $f_{n+1} \in \mathbb{R}^{\frac{H \times 2}{r_{n}} \times \frac{W \times 2}{r_{n}} \times \frac{C_{n}}{2}}$. The blocks are stacked in ConP to achieve the final output $\tilde{y} \in \mathbb{R}^{H \times W \times K}$.

\subsection{DivP}
DivP estimates the segmentation mask of each rater from the calibrated mask of ConP. The input of DivP is the estimated segmentation mask $\tilde{y}$. The output of DivP is multi-rater segmentation masks $\tilde{z}^{[M]}$. DivP is implemented by stacked multi-head attention blocks. The final multi-head attention block has $M$ heads, where each head estimates one rater's segmentation annotation.

Consider the first block in the stack.
The multi-head self-attention $ {\rm MHA} (\tilde{y} + E_{y},\tilde{y} + E_{y}, \tilde{y})$ is first applied on $\tilde{y} \in \mathbb{R}^{H \times W \times K}$ estimated by ConP.
Like which in ConP, ${\rm MHA} (query, key, value)$ denotes multi-head attention mechanism, $E_{y}$ are positional encodings. We reshape the feature maps into a sequence of flattened patches before the attention. Different from ConP, we adopt self-attention strategy in DivP, which means $query$, $key$ and $value$ are set as the same when calculating dot-product attention. 
Residual connection, normalization and MLP layer are applied following each head to facilitate training. We stack four of such blocks in DivP. The final estimated multi-rater probability maps are used to calculate the estimated confidences $\tilde{w}^{[M]}$. 

In order to integrate the confidences $\tilde{w}^{[M]}$ into ConP attentive blocks, we use separable convolution (\cite{chollet2017xception} to embed these maps to the same size as the target segmentation features in ConP. Separable convolution contains a pair of point-wise convolution and depth-wise convolution for the embedding. Point-wise convolution keeps the scale of the maps but deepen the channels, while depth-wise convolution downsamples the features but keeps the channel number. These layers can not only reshape the maps but can also transfer the maps to the deep features for the integration.

\subsection{Supervision}
\begin{figure}[!t]
\centering
\includegraphics[width=0.8\textwidth]{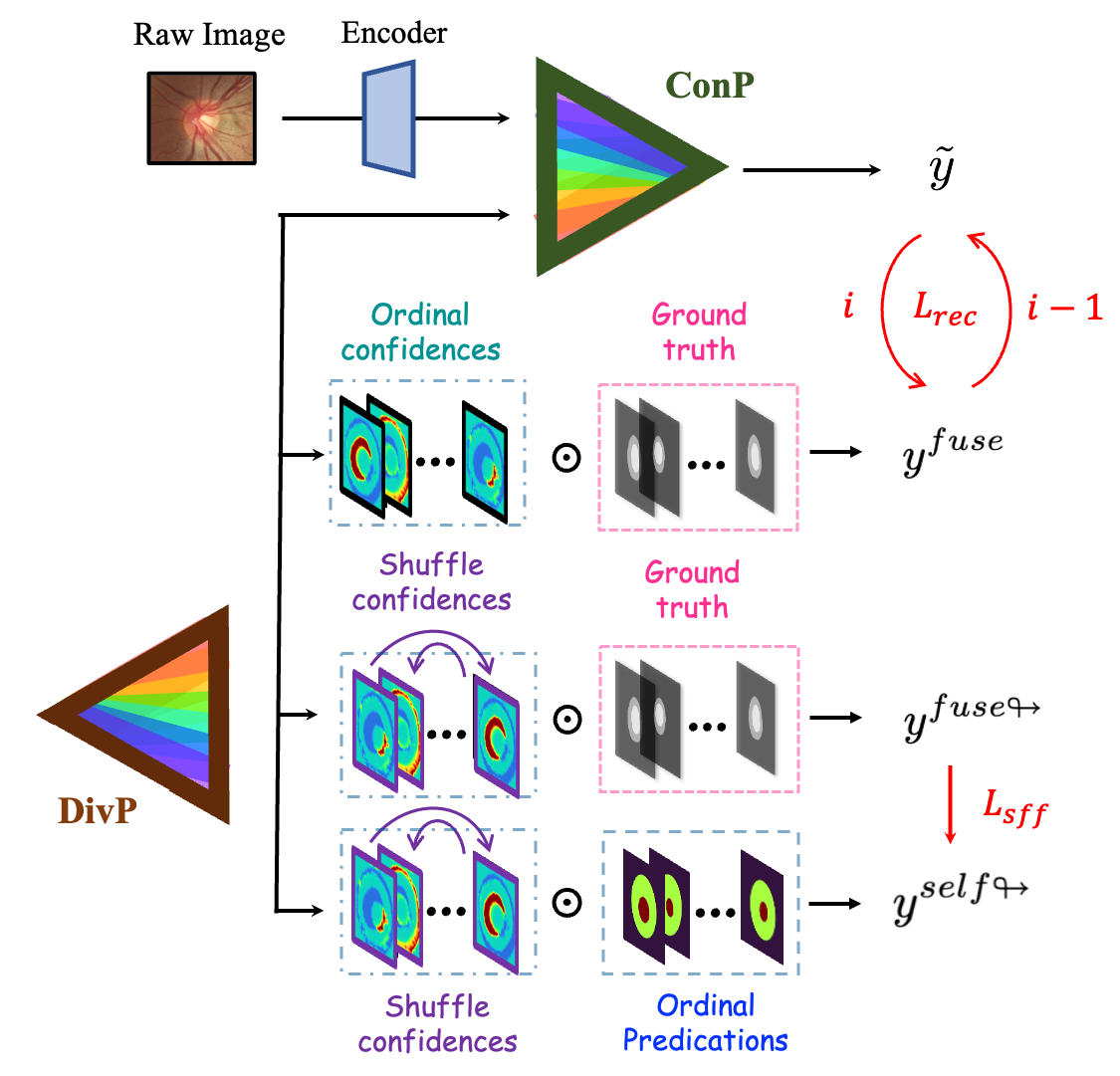}
\caption{Supervision of MrPrism. The overall supervision consists of recurrence loss $\mathcal{L}_{rec}$ and shuffle loss $\mathcal{L}_{sff}$.}
\label{fig:supervision}
\end{figure}


\subsubsection{Recurrence Supervision}

Based on the half-quadratic optimization, the two sub-problems are partially constrained by the intermediate results of the other in the recurrence.
Following this general guidance, we design the recurrence loss for the supervision of MrPrism. In recurrence loss, ConP and DivP are supervised by the other's previous predictions. Specifically, consider the supervision of the $i^{th}$ recurrence. ConP is supervised by the self-fusion label $y^{self}$  (shown in Eqn.  (\ref{equ:gt})) calculated from DivP. It constraints ConP to predict the calibrated segmentation following the given multi-rater confidences. DivP is supervised by the calibrated segmentation mask $\tilde{y}$ provided by ConP. The supervision is imposed on the self-fusion label $y^{self}$ of DivP. It constraints DivP to predict the correct multi-rater confidences so that the self-fusion mask can reconstruct the provided calibrated segmentation. An illustration of the loss functions is shown in Fig. \ref{fig:supervision}.
Formally, $\mathcal{L}_{rec}$ for the $i^{th}$ recurrence is represented as:
\begin{equation}\label{equation:la}
\begin{split}
     \mathcal{L}_{rec}^{i} & = L_{ssim}  (\tilde{y}_{i}, y^{self}_{i-1}) + L_{ssim} (y^{self}_{i}, \tilde{y}_{i}),
\end{split}
\end{equation}
where $L_{ssim}$ is SSIM loss function. The first term is the supervision of ConP and the second term is that of DivP. Note that ConP is supposed to learn the raw image prior $p$ that cannot be supervised. Therefore, instead of pixel-level supervision, we adopt Structural Similarity Index  (SSIM) as the loss function. SSIM constraints the estimated mask to have a similar structure to the self-fusion label but allows the slight difference caused by the raw image prior. 

\begin{algorithm}[h]
\caption{MrPrism Algorithm}\label{alg:mrprism}
\begin{algorithmic}[1]
\State Given the segmentation encoder with the parameters $\theta_{e}$, ConP with the parameters $\theta_{c}$ and DivP with the parameters $\theta_{d}$. 

\State Let $\lambda$ denotes the learning rate, $\zeta$ denotes the weight of shuffle loss, $\tau$ denotes the times of recurrence. $\mathcal{T}_{0.5}$ denotes the threshold of 0.5, $L_{ssim}$ and $L_{ce}$ denote SSIM and cross-entropy loss functions, respectively.

\While{$Training$}
    \State Sample $ (x, z^{[M]})$ from training dataset
    \State Initialize $\tilde{w}_{0}^{[M]}$ following Eqn.  (\ref{equation:MCA})
    \State Initialize $y_{0}^{fuse}$ following Eqn.  (\ref{equation:MCA2})
    \State $f_{0} \gets {\rm Encoder} (x)$
    \For{$i$ in range 1 to $\tau$}  \Comment{Recurrence process}
    \State $\tilde{y}_{i} \gets {\rm ConP} (f_{0}, w_{i-1}^{[M]})$
    \State $\tilde{z}_{i}^{[M]} \gets {\rm DivP} (\tilde{y}_{i})$
    \State $\tilde{w}_{i}^{[M]} \gets {\rm log} \; \tilde{z}_{i}^{[M]}$
    \State $y^{fuse}_{i} \gets \tilde{z}_{i}^{[M]} \odot z^{[M]} $
    \State Shuffle $\tilde{z}_{i}^{[M]}$ to $\tilde{z}_{i}^{[M] \looparrowright}$
    \State  $y^{fuse\looparrowright}_{i} \gets \tilde{z}^{[M] \looparrowright}_{i} \odot z^{[M]} $
    
    \State $y_{i}^{self\looparrowright} \gets \tilde{z}_{i}^{[M] \looparrowright} \odot \mathcal{T}_{0.5}  (\tilde{z}^{[M]})$
    
    \State $\mathcal{L}_{rec}^{i} \gets L_{ssim} (\tilde{y}_{i}, y^{fuse}_{i-1}) + L_{ssim} (y^{fuse}_{i}, \tilde{y}_{i})$
    \State $\mathcal{L}_{sff}^{i} \gets L_{ce}  (y^{self\looparrowright}_{i},y^{fuse\looparrowright}_{i})$
    \EndFor
    \State $\theta_{e} \gets \theta_{e} + \lambda \sum_{i=1}^{\tau} \frac{\partial \mathcal{L}_{rec}^{i}}{\partial \theta_{e}}$
    \State $\theta_{c} \gets \theta_{c} + \lambda \sum_{i=1}^{\tau} \frac{\partial \mathcal{L}_{rec}^{i}}{\partial \theta_{c}}$
    \State $\theta_{d} \gets \theta_{d} + \lambda \sum_{i=1}^{\tau}  (\frac{\partial \mathcal{L}_{rec}^{i}}{\partial \theta_{d}} + \zeta  \frac{\partial \mathcal{L}_{sff}^{i}}{\partial \theta_{d}})$
\EndWhile

\end{algorithmic}
\end{algorithm}

\subsubsection{Shuffle the Multi-head Supervision}\label{sec:shuffle}
According to Eqn.  (\ref{equation:HQ})-\(\raisebox{.5pt}{\textcircled{\raisebox{-.9pt} {1}}}\), besides the recurrence loss, DivP is also supervised by the multi-rater labels. A basic implementation is to supervise each head of DivP by the corresponding multi-rater label, i.e., multi-head  (MH) loss function, which can be represented as:
\begin{equation}\label{equation:mh}
    \mathcal{L}_{MH} = \sum_{m = 1}^{M} L_{ce}  (\tilde{z}^{m},z^{m}),
\end{equation}
where $L_{ce}$ denotes the cross-entropy loss function. However, since the multi-rater confidence and segmentation are derived from the same estimation, the recurrence is easy to fall into the trivial solution. To avoid it, we disentangle the multi-rater confidences and multi-rater segmentation masks by the shuffling. Since under any multi-rater confidences, the fusion of ground-truth labels $z^{[M]}$ and predicated labels $\tilde{z}^{[M]}$ should be the same, thus the shuffled self-fusions of them can be used for the supervision. An illustration of the shuffle supervision is shown in Fig. \ref{fig:supervision}. Specifically, the shuffle supervision of DivP can be represented as:
\begin{equation}\label{equation:sff}
    \mathcal{L}_{sff} = L_{ce}  (y^{self\looparrowright},y^{fuse\looparrowright}),
\end{equation}
where,\\
\begin{align}\label{eqn:sffself}
    y^{fuse\looparrowright} &= \tilde{z}^{[M] \looparrowright}_{i} \odot z^{[M]} \\
    & = softmax (\sum_{m = 1}^{M} log (\tilde{z}^{m_{\looparrowright}}) \cdot z^{m} + p_{u}), \nonumber
\end{align}
and,
\begin{align}\label{eqn:sffselfes}
    y^{self\looparrowright} & = \tilde{z}_{i}^{[M] \looparrowright} \odot \mathcal{T}_{0.5}  (\tilde{z}^{[M]})\\
    & = softmax (\sum_{m = 1}^{M} log (\tilde{z}^{m_{\looparrowright}}) \cdot \mathcal{T}_{0.5}  (\tilde{z}^{m}) + p_{u}). \nonumber
\end{align}
In Eqn.  (\ref{equation:sff}), $L_{ce}$ denotes the cross-entropy loss function.In Eqn.  (\ref{eqn:sffselfes}), $\mathcal{T}_{0.5}$ denotes the threshold of 0.5.
In practice, we empirically shuffle three times in each recurrence for the supervision.

Another interesting finding is that the ablation study shows $\mathcal{L}_{MH}$ is no longer needed after applying shuffle supervision: the overall performance was even improved without $\mathcal{L}_{MH}$. A possible explanation is that estimating multi-rater labels from the calibrated segmentation mask is a multi-solution problem. Many different separations are possible for one calibrated mask. Thus the supervision imposed on each of the heads may take extra constraints in the training, and thus negatively affects the performance.


\subsubsection{Overall Supervision}
Consider ConP and DivP run once each, i.e, from $f_{0}$ to $\tilde{z}^{[M]}$, as one recurrence. Each instance will run $\tau$ recurrences in a single epoch. We backforward the gradients of the model after $\tau$ times of recurrence. The total loss function is represented as:
\begin{equation}\label{equation:totalloss}
    \mathcal{L}_{total} = \sum_{i = 1}^{\tau} \mathcal{L}_{rec}^{i} + \zeta \; \mathcal{L}_{sff}^{i},
\end{equation}
where $\zeta$ is the weight hyper-parameter. The gradients are back-forward individually in each recurrence, which means the gradients of the $i^{th}$ recurrence will not affect the $ (i-1)^{th}$ recurrence.

\section{Experiment}
\subsection{Datasets}
Extensive experiments are conducted to verify the effectiveness of the proposed framework on five different types of medical segmentation tasks with data from varied image modalities, including color fundus images, CT and MRI.

\textbf{OD/OC Segmentation} The experiments of OD/OC segmentation from fundus images are conducted on two public released benchmarks, REFUGE and RIGA. REFUGE (\cite{orlando2020refuge} is a publicly available dataset for optic cup and disc (OD/OC) segmentation and glaucoma classification, which contains in total 1200 color fundus images
Seven glaucoma experts from different organizations first labeled the optic cup and disc contour masks manually, and a senior expert with the seven graders together arbitrated the final ground-truths for the validation and testing. 
RIGA benchmark (\cite{almazroa2017agreement} is a publicly available dataset for OD/OC segmentation, which contains in total 750 color fundus images. Six glaucoma experts from different organizations labeled the optic cup and disc contour masks manually. We select 655 samples in it for the training and 95 heterologous samples for the testing. A senior glaucoma expert with more than ten years' experience was invited to arbitrate the annotations of the test set.

\textbf{Brain-Tumor Segmentation} The brain-tumor segmentation from MRI images are conducted on a blend dataset of QU-BraTS 2020(\cite{mehta2021qu} and QUBIQ-BrainTumor(\cite{qubiq} subset. The dataset contains 369 cases for training, 153 cases for validation and 170 cases for testing. The training set and test set of QUBIQ are added to the validation set and test set respectively. The cases are annotated by two to four raters and arbitrated by two senior clinical experts. 

\textbf{Prostate Segmentation} The prostate segmentation from MRI images are conducted on QUBIQ-prostate(\cite{qubiq} subset. The prostate segmentation contains two subtasks. The dataset contains 33 cases for training and 15 cases for testing. Seven different raters annotated the images. The test cases are arbitrated by a senior clinical expert. 

\textbf{Brain-Growth Segmentation} The brain-growth segmentation from MRI images are conducted on QUBIQ-BrainGrowth(\cite{qubiq} subset. The dataset contains 34 cases for training and 5 cases for testing. Seven different raters annotated the images. The test cases are arbitrated by a senior clinical expert.

\textbf{Kidney Segmentation} The kidney segmentation from CT images are conducted on QUBIQ-kidney(\cite{qubiq} subset. The dataset contains 20 cases for training and 4 cases for testing. Three different raters annotated the images. The test cases are arbitrated by a senior clinical expert.

\begin{table*}[h]
\caption{Quantitative comparison between MrPrism and SOTA multi-rater learning strategies. Rec $i$ indicates the results of the $i^{th}$ recurrence.}
\centering
\resizebox{\columnwidth}{!}{
\begin{tabular}{c|ccccccccc}
\hline
                            & \multicolumn{2}{c}{OD/OC  (REFUGE)}          & \multicolumn{2}{c}{OD/OC  (RIGA)}            & Brain-Growth         & Brain-Tumor          & Prostate1             &
                        Prostate2           &    
                            Kidney               \\ \cline{2-10} 
                            & $\mathcal{D}_{disc}$                   & $\mathcal{D}_{cup}$                   & $\mathcal{D}_{disc}$                   & $\mathcal{D}_{cup}$                   & $\mathcal{D}_{brain}$                   & $\mathcal{D}_{tumor}$                   & $\mathcal{D}_{pros1}$                   & $\mathcal{D}_{pros2}$  & $\mathcal{D}_{kidney}$                   \\ \hline
WDNet (\cite{guan2018said}                         & 95.72                & 84.16                & 96.43                & 81.55                & 83.13                & 84.22                & 84.62                & 73.65 & 72.50                \\
CL (\cite{rodrigues2018deep}                        & 95.31                & 83.41                & 95.49                & 81.11                & 81.56                & 77.34                & 85.59                & 73.48  & 74.36                \\
AggNet (\cite{albarqouni2016aggnet}                      & 94.96                & 83.66                & 95.38                & 82.27                & 82.72                & 83.37                & 78.33                & 71.77 & 71.39                \\
CM (\cite{tanno2019learning}                          & 94.50                & 84.72                & 96.56                & 83.11                & 80.56                & 79.21                & 85.91                & 74.55 & 74.85                \\
MaxMig (\cite{cao2019max}                     & 95.72                & 84.45                & 96.64                & 84.81                & 84.62                & 86.41                & 86.27                & 74.06 & 74.68                \\
MRNet (\cite{ji2021learning}                      &  94.75                   &      85.63                 &        96.27              & 85.82                     &    83.67                  &  87.63                    & 87.04                      & 75.43 & 75.19                     \\
Diag (\cite{wu2022opinions}                        & 95.17                     &   86.23                   &      96.74                &    86.17                  & -                    & -                     & -                     & -  & -                    \\ \hline
MrPrism-Rec0                        & 94.62                     &   84.33                   &      96.18                &    85.28                  &   83.53                   &  87.81                    & 86.17                     & 75.20 & 75.24                     \\
MrPrism-Rec1                        & 95.36                     &   87.59                   &      96.37                &    87.42                  &   84.25                   & 88.23                     & 87.35                     & 76.88 & 76.22                     \\
MrPrism-Rec2                        & 95.67                     &   \textbf{88.74}                   &      \textbf{96.66}                &    88.15                  &   85.50                   &  \textbf{89.58}                    &  \textbf{88.45}                    &  77.29 & \textbf{76.51}                    \\
\multicolumn{1}{c|}{MrPrism-Rec3} & \multicolumn{1}{c}{\textbf{95.68}} & \multicolumn{1}{c}{88.53} & \multicolumn{1}{c}{96.62} & \multicolumn{1}{c}{\textbf{88.30}} & \multicolumn{1}{c}{\textbf{85.52}} & \multicolumn{1}{c}{89.44} & \multicolumn{1}{c}{88.37} & \multicolumn{1}{c}{\textbf{77.34}} & \multicolumn{1}{c}{76.48}\\ \hline
\end{tabular}}\label{tab:overall}
\vspace{-8pt}
\end{table*}

\begin{figure}[h]
\centering
\includegraphics[width=0.8\textwidth]{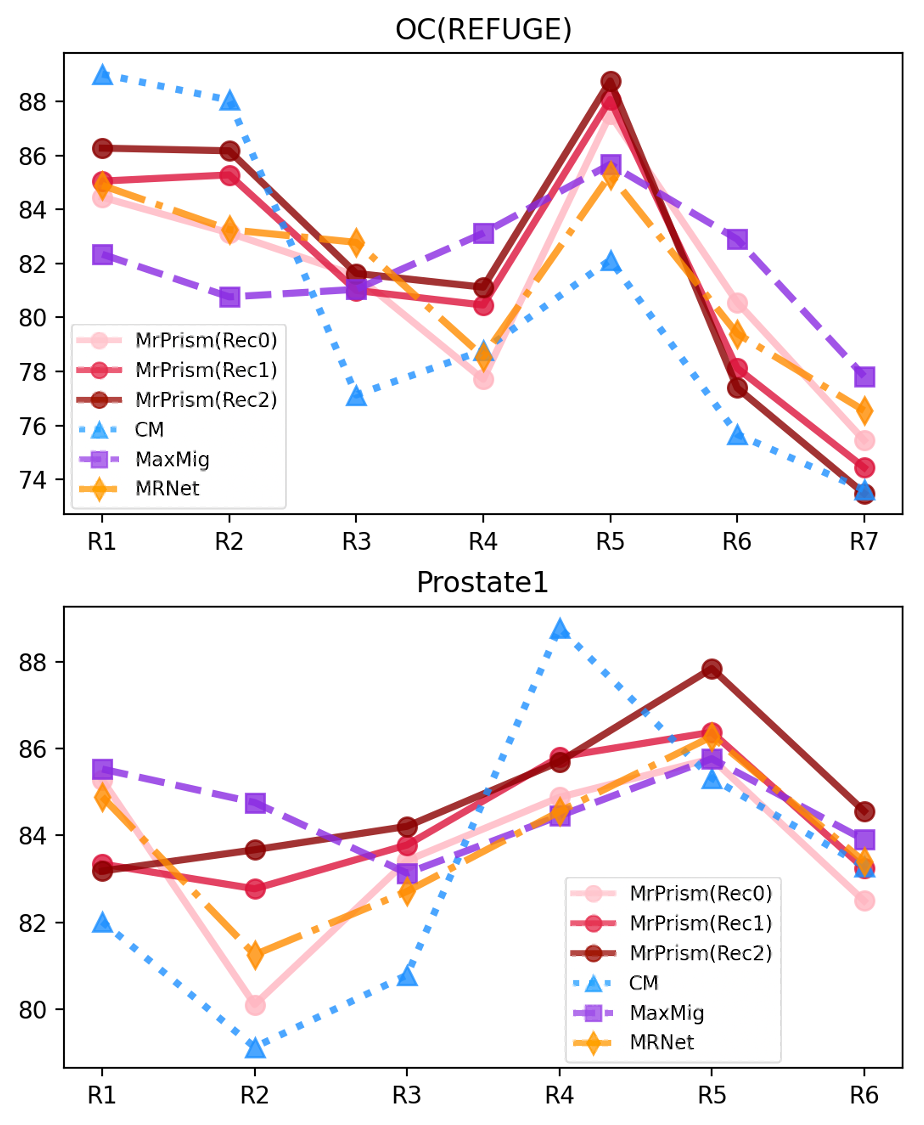}
\caption{The comparison of different methods on each of the multi-rater labels measured by dice coefficient (\%). R $i$ denotes the annotation of the rater $i$.}
\label{fig:exprec}
\vspace{-20pt}
\end{figure}

\subsection{Experimental Setup}
\subsubsection{Implement Details}

All the experiments are implemented with the PyTorch platform and trained/tested on 4 Tesla P40 GPU with 24GB of memory. All training and test images are uniformly resized to the dimension of 256$\times$256 pixels. ConP is stacked with five stages with the patch sizes 8, 7, 7, 7, 7 in the attentive mechanism. DivP is stacked with four stages with patch sizes 7, 7, 7, 7, and 4, 4, 4, $M$ heads in the multi-head attention mechanism. The weight constant of shuffle loss $\zeta$ is set as 0.3. Generally, each instance will run three times of recurrence in the training and inference stages. The proposed model is trained using Adam optimizer  (\cite{kingma2014adam} for 150 epochs. The learning rate is always set to 1 $\times 10^{-4}$. 
We initialize the recurrence by assigning the equal confidence to each rater. Formally, we set initial $w^{m}_{0}$ as:
\begin{equation}\label{equation:MCA}
    w^{m}_{0} = \rm{log} \; [\psi ^{m} + \varepsilon^{m}]_{0}^{1}
\end{equation}
where $\psi^{m}$ decide the average of atlas $w^{m}_{0}$, $\varepsilon^{m}$ adds different perturbation on each element. $[\cdot]_{0}^{1}$ denotes the value is clipped to range of 0 to 1. For the given $w^{m}$, the ground truth is generated by 
\begin{equation}\label{equation:MCA2}
y_{0}^{self} = softmax (\sum_{m = 1}^{M} w_{0}^{m} \cdot z^{m}_{n} + p_{u}).
\end{equation}
We sample $\psi ^{[M]}$ from uniform distribution, $\psi ^{[M]} \thicksim U (a,b)$. $\varepsilon^{m}$ is sampled from a normal distribution, $\varepsilon^{m} \thicksim \mathcal{N} (\mu, \sigma^2)$, where $a,b,\mu,\sigma$ are set as 0.1, 0.9, 0, 0.2 respectively. The confidence maps will be set as the same for each batch of the data. 

\subsubsection{Evaluation Metric}
The segmentation performance is evaluated by soft dice coefficient  ($\mathcal{D}$) through multiple threshold level, set as  (0.1, 0.3, 0.5, 0.7, 0.9). At each threshold level, the predicted probability map and soft ground-truth are binarized with the given threshold, and then the dice metric  (\cite{milletari2016v} is computed. $\mathcal{D}$ scores are obtained as the averages of multiple thresholds.

\vspace{-8pt}
\subsection{Experiment Results}
\subsubsection{Overall Performance}
To verify the self-calibrated segmentation performance of the proposed model, we compare MrPrism with SOTA multi-rater learning methods. The selected methods include WDNet (\cite{guan2018said}, CL (\cite{rodrigues2018deep}, CM  (\cite{tanno2019learning}, AggNet (\cite{albarqouni2016aggnet}, MaxMig (\cite{cao2019max}, MRNet (\cite{ji2021learning} and Diag  (\cite{wu2022opinions}. To be specific, CL, CM, AggNet, MaxMig, and Diag jointly learn the prediction and multi-rater confidences assignment, in which CL, CM, and MaxMig implicitly learn the confidences while AggNet and Diag explicitly learn the confidences. 
WDNet and MRNet are calibrated segmentation methods that need the rater confidence provided. Since the multi-rater confidences are not available in our scenario, all raters are considered equal when applying these methods. Diag is a method that uses the disease diagnosis/classification performance as a standard to evaluate the confidences of multi-rater segmentation labels. In OD/OC segmentation, we use glaucoma diagnosis as the standard. In the other segmentation tasks, no diagnosis task is associated, so the results of Diag are not reported. The experiments are conducted on a wide range of medical segmentation tasks, including OD/OC segmentation, infant brain growth segmentation, brain tumor segmentation, prostate segmentation, and kidney segmentation. In MrPrism, the predictions of ConP in four recurrences are used for the comparison, which are denoted as Rec0, Rec1, Rec2 and Rec3. The detailed quantitative results are shown in Table \ref{tab:overall}. 

\begin{figure*}[h]
\centering
\includegraphics[width=0.85\textwidth]{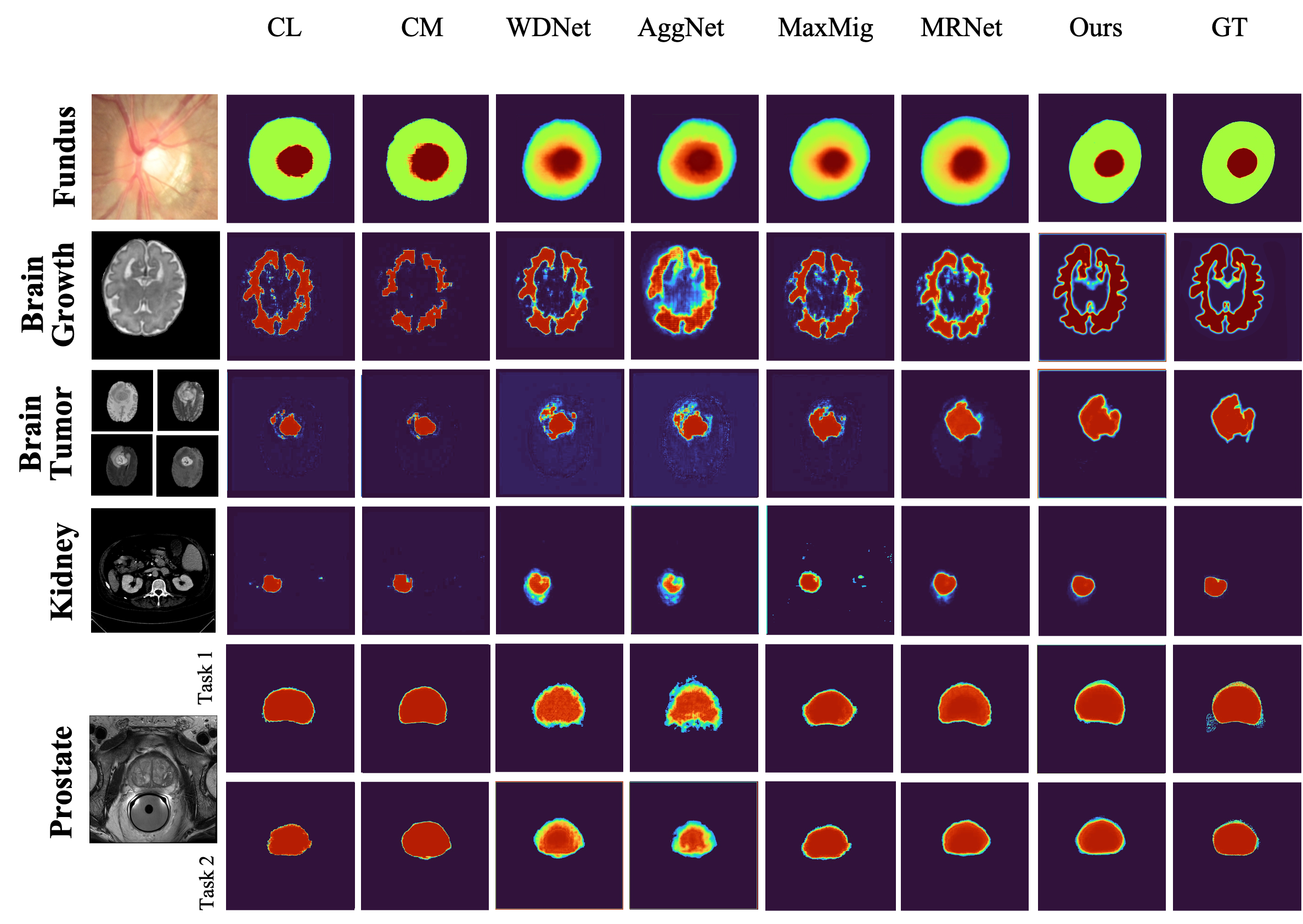}
\caption{Segmentation results of different strategies for five different medical segmentation tasks.}
\label{fig:expvis}
\end{figure*}

\begin{figure*}[h]
\centering
\includegraphics[width=0.85\textwidth]{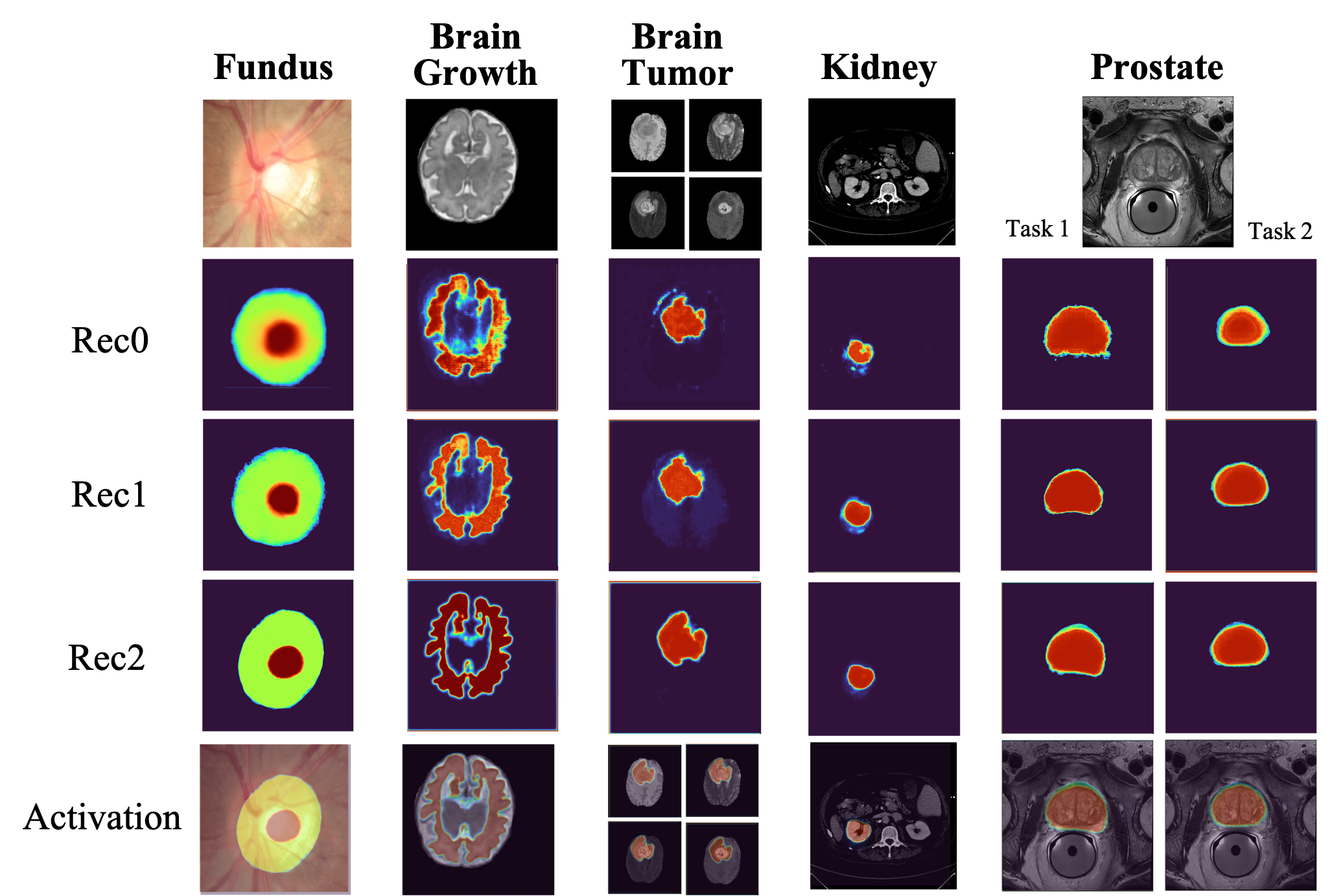}
\caption{Segmentation results of different strategies for five different medical segmentation tasks.}
\label{fig:recvis}
\end{figure*}

As listed in Table \ref{tab:overall}, the proposed MrPrism achieves superior performance on all tasks compared with SOTA multi-rater learning methods. These quantitative results verify that the self-calibrated segmentation results produced by MrPrism can represent the inter-rater agreement more accurately than the other methods. In addition, it is worth noting that the proposed MrPrism outperforms the other methods by a large margin when the inter-rater variety is large, e.g., on optic cup and brain-tumor. 
The superior performance of MrPrism on these tasks demonstrates the effectiveness of the proposed framework, which is tailored for medical image segmentation with ambiguous annotations, by taking advantage of the iterative self-calibrated segmentation.

Comparing the results in different recurrences of MrPrism in Table \ref{tab:overall}, we can see the performance of MrPrism keeps increasing on most of the metrics. The improvement is significant in the earlier recurrences and gradually diminishes in the later ones. 
The performance becomes stable in the fourth recurrence, indicating the results have converged. In practice, we generally recurrent three times for efficiency. It is also worth noting that even though the large inter-rater variety initially causes a lower performance, MrPrism can also substantially calibrate the results through the recurrence and gain high performance. For example, on optic cup segmentation which the inter-rater variety is large, MrPrism gets 85.28\% on RIGA dataset in the first recurrence but increases 3.02\% in four recurrences, and finally gains a high 88.30\%. These experimental results demonstrate that MrPrism is able to calibrate the results by itself through the recurrence.

We show a visualized comparison of the methods in Fig. \ref{fig:expvis}. We can see CL and CM, which implicitly learn the multi-rater confidences, are inclined to be overconfident in the inaccurate results, while AggNet and MaxMig, which explicitly learns the multi-rater confidences are prone to obtain ambiguous results. WDNet and MRNet need the correct multi-rater confidences to be provided, thus to be ambiguous under the default majority vote assumption. As a comparison, the proposed self-calibrated segmentation can estimate the result from uncertain to confident through the recurrences, finally achieving a confident and calibrated result.


In order to further explore the discrepancies of these methods, we compare the performance of the methods on multi-rater annotations respectively. The results are shown in Fig. \ref{fig:exprec}. The experiment is conducted on two different segmentation tasks: optic cup segmentation  (OC-REFUGE), in which the inter-observer variety is large, and prostate segmentation  (Prostate1), in which the inter-observer variety is relatively small. Several representative multi-rater learning strategies are selected for the comparison, including the three recurrences of MrPrism, CM, MaxMig, and MRNet. Compared with these methods, we can see CM, which implicitly learns the multi-rater confidences is inclined to believe few dominant raters, e.g., R1 and R2 in OC-REFUGE or R4 in Prostate1. However, without the dynamic calibration, the dominant raters are often incorrect, e.g., the credible rater in OC-REFUGE is actually R5 but not R1 or R2. It leads to its confident but incorrect predictions. MaxMig, which explicitly learns the multi-rater confidences is prone to take a balance between the multiple raters. However, this strategy does not work well when the inter-rater variety is large, e.g., MaxMig gains only 84.45\% on OC-REFUGE, while MrPrism gains 88.74\%. MRNet requires the multi-rater confidences provided for the prediction. Under the default majority vote setting, its performance is close to the uncalibrated MrPrism-Rec0. Different from the existed strategies, MrPrism shows to be more adaptable to different tasks. On OC-REFUGE where few raters are much better than the others, MrPrism keeps increasing on some of the raters in the recurrences, e.g., R1, R2, R4, R5, while dropping on the other rater, e.g., R6 R7. It indicates that MrPrsim is able to discriminate the better raters from the others in the recurrences under the constraint of the raw image prior. While on Prostate1 where most raters are equally correct, MrPrism raises the overall performance of most raters in the recurrences but will not be overly biased like CM.

\begin{table*}[h]
\centering
\setlength\arrayrulewidth{0.1pt}
\caption{Quantitative comparison between ConP and SOTA calibrated/non-calibrated segmentation methods.}
\centering
\resizebox{\columnwidth}{!}{
\begin{tabular}{c|c|cc|cc|c|c|c|c|c}
\toprule
     &  & \multicolumn{2}{c|}{OD/OC (REFUGE)} & \multicolumn{2}{c|}{OD/OC (RIGA)} & \multicolumn{1}{c|}{Brain-Growth} & \multicolumn{1}{c|}{Brain-Tumor} & \multicolumn{1}{c|}{Prostate1} &
     \multicolumn{1}{c|}{Prostate2} &
     \multicolumn{1}{c}{Kidney}\\ \hline
      &     & $\mathcal{D}_{disc}$     & $\mathcal{D}_{cup}$      & $\mathcal{D}_{disc}$       & $\mathcal{D}_{cup}$          & $\mathcal{D}_{brain}$                   & $\mathcal{D}_{tumor}$                   & $\mathcal{D}_{pros1}$                   & $\mathcal{D}_{pros2}$  & $\mathcal{D}_{kidney}$ \\ \hline
\multirow{3}{*}{No Calibrated } & AGNet   & 90.21  & 71.86  & 96.31  & 78.05     & 79.05  & 80.38  & 84.67 & 69.72 & 70.69 \\
 & pOSAL     & 94.52  & 83.81  & 95.85  & 84.07     & 81.56  & 83.78 & 85.46  & 71.81 & 71.60   \\ 
 &BEAL     & 94.84  & 84.92 & 97.08  & 85.97     & 82.79  & 85.70  & 85.96  & 73.51 & 72.41   \\ \hline
 \multirow{3}{*}{Calibrated} & LUW    & 94.72  & 85.12  & 96.75  & 86.63     & 83.67  & 87.44          & 86.31  & 75.78 & 73.25   \\
 & UECNN    & 94.42  & 84.80  & 96.37  & 85.52     & 82.77  & 85.81  & 86.30  & 70.32 & 72.28   \\
 & MRNet    & 95.00  & 86.40  & 97.55  & 87.20     & 84.35  & 88.41 & 87.22   & 75.94 & 74.96   \\ \cline{1-11}
 Self-calibrated & ConP (ours) & \textbf{96.02}  & \textbf{88.79}  & \textbf{97.84}   & \textbf{89.35}     & \textbf{86.33}  & \textbf{89.62}            & \textbf{87.89}  & \textbf{77.45} & \textbf{75.50}  \\ \hline
\bottomrule
\end{tabular}%
}
\label{tab:conp}
\end{table*}

\begin{table*}[h]
\centering
\setlength\arrayrulewidth{0.1pt}
\caption{Quantitative comparison between DivP and SOTA label-fusion methods.}
\centering
\resizebox{\columnwidth}{!}{
\begin{tabular}{c|cc|cc|c|c|c|c|c}
\toprule
      & \multicolumn{2}{c|}{OD/OC  (REFUGE)} & \multicolumn{2}{c|}{OD/OC  (RIGA)} & \multicolumn{1}{c|}{Brain-Growth} & \multicolumn{1}{c|}{Brain-Tumor} & \multicolumn{1}{c|}{Prostate1} &
     \multicolumn{1}{c|}{Prostate2} &
     \multicolumn{1}{c}{Kidney}\\ \hline
          & $\mathcal{D}_{disc}$     & $\mathcal{D}_{cup}$      & $\mathcal{D}_{disc}$       & $\mathcal{D}_{cup}$           & $\mathcal{D}_{brain}$                   & $\mathcal{D}_{tumor}$                   & $\mathcal{D}_{pros1}$                   & $\mathcal{D}_{pros2}$  & $\mathcal{D}_{kidney}$ \\ \hline
 MV   & 94.52  & 80.31  & 95.60  & 82.47     & 83.06  & 79.60  & 89.67 & 86.10 & 84.36 \\
 STAPLE     & 95.74  & 80.91  & 95.71  & 83.75     & 85.41  & 84.35 & 89.38  & 86.80 & 85.53   \\ 
 MaxMig     & 96.43  & 87.49 & 97.80  & 88.45     & 88.78  & 89.70  & 92.85  & 90.21 & 87.62   \\
 Diag    & 96.75  & 91.22  & 97.93  & 90.79     & -  & -         & -  & - & -  \\ \cline{1-10}
 DivP (ours)   & \textbf{97.55}  & \textbf{92.76}  & \textbf{98.70}   & \textbf{91.53}     & \textbf{89.70}  & \textbf{91.01}            & \textbf{94.29}  & \textbf{92.26} & \textbf{91.12}  \\ \hline
\bottomrule
\end{tabular}%
}
\label{tab:divp}
\end{table*}

\begin{table}[h]
\centering
\setlength\arrayrulewidth{0.1pt}
\caption{Ablation study of the integration strategies and the loss functions on the brain tumor segmentation.}
\resizebox{\columnwidth}{!}{
\begin{tabular}{ccc|ccc|c}
\toprule
\multicolumn{3}{c|}{Integration}         & \multicolumn{3}{c|}{Loss} & Brain-Tumor \\ \hline
Concat & ConvLSTM  & Ours      & $\mathcal{L}_{MH}$      & $ \mathcal{L}_{sff}$ & $ \mathcal{L}_{MH} + \mathcal{L}_{sff}$    & $ \mathcal{D}_{tumor} $ \\ \hline
   \checkmark    &          &                    &        &     \checkmark   &        &     86.14      \\
       &      \checkmark    &                    &        &  \checkmark      &        & 86.73   \\
       &          & \checkmark      &        &     \checkmark   &       &      \textbf{89.58}       \\
       &          &            \checkmark        &  \checkmark      &        &        &  86.69      \\
       &          &           \checkmark         &        &        &   \checkmark     &       88.52      \\ \hline
\bottomrule
\end{tabular}}
\vspace{-10pt}
\label{tab:ablation}
\end{table}

\subsubsection{Comparing ConP with SOTA calibrated segmentation methods}
The calibrated segmentation methods are commonly proposed to capture the uncertainty from multiple labels by assuming that there is no priority on any of them. Unlike their strategies, we jointly optimize the calibrated segmentation model (ConP) with the priority evaluation (DivP), and achieve stronger calibration ability on ConP. To verify its calibration ability, we compare it with SOTA calibrated segmentation methods, including LUW (\cite{rupprecht2017learning}, UECNN (\cite{jensen2019improving},and MRNet  (\cite{ji2021learning} and SOTA non-calibrated segmentation methods, including AGNet  (\cite{acnet}, BEAL  (\cite{beal},and pOSAL (\cite{posal} on a wide range of medical segmentation tasks. The calibrated and our methods are all compared under majority vote in the inference following the setting of previous works.
The quantitative results are shown in Table \ref{tab:conp}. As listed in Table \ref{tab:conp}, calibrated methods work better than non-calibrated ones, and the proposed self-calibrated method works better than calibrated ones. Proposed ConP consistently achieves superior performance, indicating its stronger calibration ability. Compared with other calibrated methods, the decomposition process in the recurrence strategy can help ConP better realize each component of the fusion and finally achieve better calibration results.


\subsubsection{Comparing DivP with SOTA label fusion methods}
On the other side, the calibrated segmentation of ConP can also facilitate DivP. We quantitatively compare the self-fusion labels of DivP with SOTA label fusion strategies in Table \ref{tab:divp}. The compared methods include traditional majority vote (MV), STAPLE  (\cite{warfield2004simultaneous}, MaxMig (\cite{cao2019max} and Diag  (\cite{wu2022opinions}. The quantitative results are shown in Table \ref{tab:divp}. We can see that the traditional majority vote shows inferior performance. STAPLE is an advanced fusion method following the majority first strategy, which also does not work well when few raters are dominant in the annotations, like OC-REFUGE. MaxMig fuses the labels based on the raw image structural prior and obtains fair results on most tasks. Diag utilizes the diagnosis labels for the label fusion and works much better. But since it requires the diagnosis labels, the application horizon is limited. The proposed method not only fuses the labels based on the raw image structural prior, but can also dynamically adjust the fusions based on the calibration, finally achieving the best consistency with the ground-truth.

\subsubsection{Ablation Study}
We conduct detailed ablation study on the integration strategies and loss functions. The quantitative results of brain tumor segmentation are shown in Table \ref{tab:ablation}. We compare our attention-convolution hybrid integration strategy with Concat, which is concatenating the segmentation and confidence feature for the integration, and ConvLSTM, which is to use ConvLSTM (\cite{shi2015convolutional} for the integration (\cite{ji2021learning}. 
Thanks to the global and dynamic nature of the scale dot-product attentive mechanism, the proposed integration gains a considerable 2.85\% improvement compared with ConvLSTM. Then we compare the effect of traditional multi-head loss $\mathcal{L}_{MH}$ and the proposed shuffle loss $\mathcal{L}_{sff}$ as an additional constraint w.r.t DivP. We can see that $\mathcal{L}_{sff}$ outperforms $\mathcal{L}_{MH}$ by disentangling the multi-rater confidences and annotations. We also note that applying $\mathcal{L}_{sff}$ individually even works better than applying the combination of $\mathcal{L}_{MH}$ and $\mathcal{L}_{sff}$. A possible explanation is that $\mathcal{L}_{MH}$ constraints each head to estimate specific rater annotations, while different separations would be possible for one segmentation, so the extra constraint of $\mathcal{L}_{MH}$ eventually degrades the performance. The experimental results show the combination of attentive integration and shuffle loss works the best for MrPrism, which is thus adopted in our final implementation.

\section{Conclusion}
Toward learning medical segmentation from multi-rater labels, we propose a self-calibrated segmentation model to calibrate the segmentation and estimate the multi-rater confidences recurrently. In this way, the shortcomings of the two independent tasks are complemented, thus gaining mutual improvement. Extensive empirical experiments demonstrated that self-calibrated segmentation outperforms the alternative multi-rater learning methods. 

\clearpage

\bibliographystyle{model2-names.bst}\biboptions{authoryear}
\bibliography{refs}
\clearpage

\section{Supplement Material}
\subsection{NETWORK ARCHITECTURE}

\begin{table}[t!]
\centering
\vspace{-100pt}
\caption{The network architecture of MrPrism. ConP is represented as Deconv [kernel size \texttimes \ kernel size, output channels], Att-Block [patch size \texttimes \ patch size, channels], Concat \& CBR [output channels]. DivP is represented as, Conv [kernel size \texttimes \ kernel size, output channels, stride], maxpool[kernel size \texttimes \ kernel size, stride ], MHA-Block[patch size \texttimes \ patch size, channels, heads]. Separable Convolutions are represented as DW-Conv[kernel size \texttimes \ kernel size, stride], PW-Conv[output channels].}
\small
\setlength{\tabcolsep}{0pt}
\begin{tabular*}{0.9\textwidth}{@{\extracolsep{\fill}}ccc@{}}
\toprule
\multicolumn{3}{c}{\textbf{Notations}} \\
\midrule
Deconv & deconvolution\\
ATT-Block & attention $\to$ add + norm $\to$ MLP $\to$ add + norm\\ 
CBR & Conv $\to$ batchnorm $\to$ residule add\\ 
MHA-Block & multi-head attention $\to$ norm $\to$ MLP $\to$ add + norm\\
DW-Conv & depthwise convolution\\
PW-Conv & pointwise convolution\\
\midrule
\multicolumn{3}{c}{\textbf{ConP}} \\
\midrule
Stage & Stem &Output\\
\midrule
Stage1 & \splitcell{Deconv[2\texttimes2, 512]\\ Att-Block[8\texttimes8, 512]\\ Concat \& CBR[512]} 
 &
  8\texttimes8\texttimes512 \\
\midrule
Stage2 & \splitcell{Deconv[2\texttimes2, 256]\\ Att-Block[7\texttimes7, 256]\\ Concat \& CBR[256]} 
 &
  16\texttimes16\texttimes256 \\
\midrule
Stage3 & \splitcell{Deconv[2\texttimes2, 128]\\ Att-Block[8\texttimes8, 128]\\ Concat \& CBR[128]} 
 &
  32\texttimes32\texttimes128 \\
\midrule
Stage4 & \splitcell{Deconv[2\texttimes2, 64]\\ Att-Block[8\texttimes8, 64]\\ Concat \& CBR[64]} 
 & 64\texttimes64\texttimes64 \\
\midrule
Stage5 & Deconv[2\texttimes2,K]  & 128\texttimes128\texttimes K \\
\midrule
\multicolumn{3}{c}{\textbf{DivP}} \\
\midrule
Stage & Stem & Output \\
\midrule
&\splitcell{Conv[8\texttimes 8, 64, 2]\\ MaxPool[3\texttimes3,2]} & 32\texttimes32\texttimes 32\\
\midrule
Stage1 &\splitcell{MHA-Block[8\texttimes 8, 32, 8]} & 32\texttimes32\texttimes 32\\
\midrule
Stage2 &\splitcell{MHA-Block[8\texttimes 8, 64, 8]} & 32\texttimes32\texttimes 64\\
\midrule
Stage3 &\splitcell{MHA-Block[8\texttimes 8, 64, 8]} & 32\texttimes32\texttimes 64\\
\midrule
Stage4 &\splitcell{MHA-Block[8\texttimes 8, 32, 8]} & 32\texttimes32\texttimes 32\\
\midrule
Stage5 &\splitcell{MHA-Block[8\texttimes 8, K \texttimes R, R]\\  Softmax} & 32\texttimes32\texttimes K\\
\midrule
\multicolumn{3}{c}{\textbf{Separable Convs}} \\
\midrule
Stage & Stem & Output \\
\midrule
Stage1 &\splitcell{DW-Conv[5\texttimes 5, 2]\\
PW-Conv[64]\\
PW-Conv[128]\\
DW-Conv[3\texttimes 3, 2]\\
PW-Conv[256]\\
PW-Conv[512]
} & \splitcell{16\texttimes16\texttimes K\\
16\texttimes16\texttimes 64\\
16\texttimes16\texttimes 128\\
8\texttimes8\texttimes 128\\
8\texttimes8\texttimes 256\\
8\texttimes8\texttimes 512}
\\
\midrule
Stage2 &\splitcell{PW-Conv[64]\\
DW-Conv[3\texttimes 3, 2]\\
PW-Conv[128]\\
PW-Conv[256]
} & \splitcell{32\texttimes32\texttimes 64\\
16\texttimes16\texttimes 64\\
16\texttimes16\texttimes 128\\
16\texttimes16\texttimes 256}
\\
\midrule
Stage3 &\splitcell{PW-Conv[64]\\
PW-Conv[128]
} & \splitcell{32\texttimes32\texttimes 64\\
32\texttimes32\texttimes 128}
\\
\midrule
\bottomrule
\end{tabular*}\label{tab:arc}
\end{table}

\begin{figure*}[h]
\centering
\includegraphics[width=0.8\textwidth]{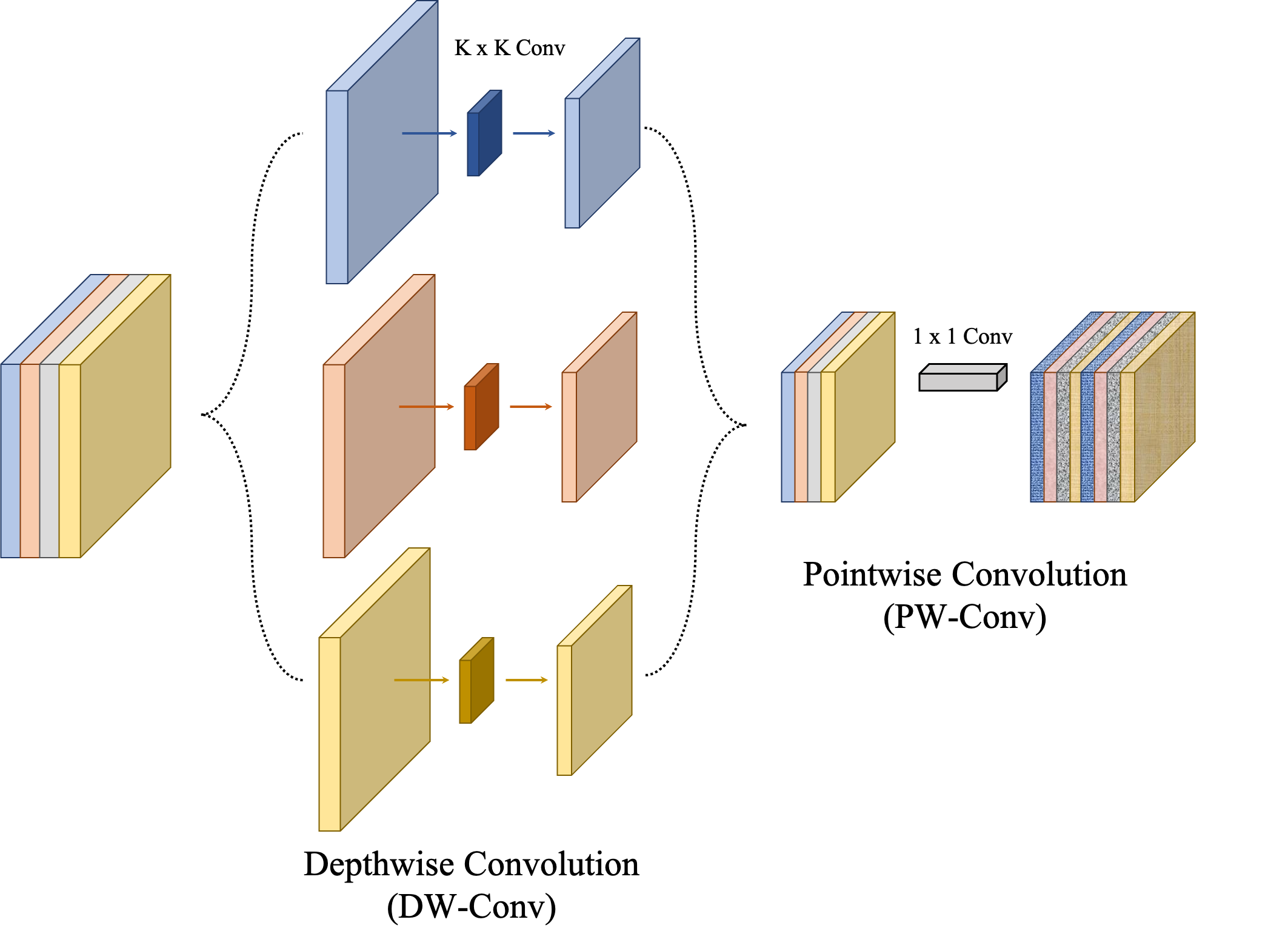}
\caption{An illustration of the separable convolution used for embeding the multi-rater confidences. DW-Conv alters the resolution. PW-Conv alters the channels.}
\label{fig:expvis}
\end{figure*}

\subsection{PROOF OF THEOREMS}
Fortunately, we find that the multi-rater confidences $w^{[M]}$ can be directly represented by $P (z^{[M]} | y^{*})$, i.e., the estimated probability map of multi-rater labels given the optimal calibrated segmentation map. This enable us to directly supervise DivP by the multi-rater labels $z^{[M]}$.

\textit{To keep clarity, follow discussions focus on an arbitrary pixel  (i,j) of the segmentation/confidence maps. Denote by a $A_{i,j} \in \mathbb{R}^{1 \times K}$ the vector at $ (i,j)^{th}$ entry of the matrix $A$. Denote by a $\upsilon^{k}\in \mathbb{R}^{1 \times K}$ the one-hot vector activated at entry $k$. The segmentation task can be considered as the classification tasks for each pixel. Consider $y$ and $y^{*}$ are drawn from random variables $Y$ and $Y^{*}$. 
}

\begin{lemma}\label{lemma:1}
Let $y^{*}$ denotes the latent accurate label. $y^{*} \thicksim Y$ is equivalent to the following statement: for every class $k \in [K]$, there exist an optimal set of parameters $w^{*[M]}$, which satisfying $ P (Y_{i,j} = \upsilon^{k} |w^{*[M]};z^{[M]};p)$ equals $P (Y^{*}_{i,j} =  \upsilon^{k} | Z^{[M]} = z^{[M]})$.
\end{lemma}
Lemma. \ref{lemma:1} is easy to see considering the pixel-wise categorical distribution $P ( Y = y | w^{*[M]}; z^{[M]}; p) = softmax (\sum_{m = 1}^{M} w^{*m} \cdot z^{m} + p).$

We then restate Proposition \ref{pro:1} with more details and prove it. 
\begin{proposition}\label{pro:1}
For every class $k,t \in [K]$, iif $w^{*m}_{i,j} (t) = {\rm log} \;P (Z^{m}_{i,j} = \upsilon^{t} | Y^{*}_{i,j} = \upsilon^{k})$, there has $ P (Y^{*}_{i,j} = \upsilon^{k} |Z^{m}_{i,j} = z^{[M]}_{i,j}) = softmax (\sum_{m = 1}^{M} w^{*m}_{i,j} \cdot z^{m}_{i,j} + p_{i,j})_{k} $.
\end{proposition}

\begin{proof}
\begin{equation}\nonumber
\begin{split}
     &{\rm log} \;P (Y^{*}_{i,j} = \upsilon^{k} | Z^{[M]}_{i,j} = z^{[M]}_{i,j})\\
     =& {\rm log} \;P (Z^{[M]}_{i,j} = z^{[M]}_{i,j} | Y^{*}_{i,j} = \upsilon^{k}) \\ 
      &+ {\rm log} \; P (Y^{*}_{i,j} = \upsilon^{k})) - {\rm log} \;P (Z^{[M]}_{i,j} = z^{[M]}_{i,j})\\
     =& \sum_{m=1}^{M} {\rm log} \;P (Z^{m}_{i,j} = z^{m}_{i,j} | Y^{*}_{i,j} = \upsilon^{k}) \\ 
      &+ {\rm log} \; P (Y^{*}_{i,j} = \upsilon^{k}) - {\rm log} \;P (Z^{[M]}_{i,j} = z^{[M]}_{i,j})\\
\end{split}
\end{equation}\nonumber
     Denote ${\rm log} \; P (Y^{*}_{i,j} = \upsilon^{k})$ by $b_{i,j}$ and ${\rm log} \;P (Z^{[M]}_{i,j} = z^{[M]}_{i,j})$ by $c_{i,j}$, respectively.
\begin{equation}
\begin{split}
     =& \sum_{m=1}^{M} w^{*m}_{i,j} \cdot z^{m}_{i,j} + b_{i,j} - c_{i,j}\\
     =& {\rm log} \;e^{\sum_{m=1}^{M} w^{*m}_{i,j} \cdot z^{m}_{i,j} + b_{i,j} - c_{i,j}}\\
     & since \sum_{k} P (Y^{*}_{i,j} = \upsilon^{k} | Z^{[M]}_{i,j} = z^{[M]}_{i,j}) = 1\\
     =& {\rm log} \;\frac{e^{\sum_{m=1}^{M} w^{*m}_{i,j} \cdot z^{m}_{i,j} + b_{i,j} - c_{i,j}}}{\sum_{k} e^{{\rm log} \; ( P ( (Y^{*})_{i,j} = \upsilon^{k} |  (Z^{[M]})_{i,j} =  (z^{[M]})_{i,j}))}}\\
     =& {\rm log} \;\frac{e^{\sum_{m=1}^{M} w^{*m}_{i,j} \cdot z^{m}_{i,j} + b_{i,j}}}{\sum_{k}e^{\sum_{m=1}^{M} w^{*m}_{i,j} \cdot z^{m}_{i,j} + b_{i,j}}}\\
     =& {\rm log} \;softmax (\sum_{m=1}^{M} w^{*m}_{i,j} \cdot z^{m}_{i,j} + b_{i,j})_{k}\\
\end{split}
\end{equation}
    Let $p_{i,j} = {\rm log} \; P (Y^{*}_{i,j} = \upsilon^{k})$, it has $ P (Y^{*}_{i,j} = \upsilon^{k} | Z^{[M]} = z^{[M]}) = softmax (\sum_{m = 1}^{M} w^{*m} \cdot z^{m} + p)_{k} $. And it is obvious $w^{*m}_{i,j} (t) = {\rm log} \;P (Z^{m}_{i,j} = \upsilon^{t} | Y^{*}_{i,j} = \upsilon^{k})$ is the unique solution of the equation. According to Lemma \ref{lemma:1}, it is also the necessary condition of $y^{*} \sim Y$. 
\end{proof}

\begin{proposition}\label{pro:2-prove}
Solving Eqn. (\ref{equation:HQ}) is equivalent to solve the following equation:
\begin{equation}
\left\{\begin{array}{ll}
\begin{split}
\  \mathcal{Y}^{'}_{i}= &\arg min_{\mathcal{Y}^{'}_{i}} \frac{\beta}{2}\Vert \mathcal{Y}^{fuse}_{i-1} - \mathcal{Y}^{'}_{i}\Vert^{2}_{2} 
\\&+ \eta \; \mathcal{P}_{x}(\mathcal{Y}^{'}_{i}) & {\raisebox{.5pt}{\textcircled{\raisebox{-.9pt} {1}}}}
\\
\  \mathcal{Y}^{fuse}_{i} = &\arg min_{\mathcal{Y}^{fuse}_{i}} \frac{\beta + 1}{2} \Vert \mathcal{Y}^{fuse}_{i} - \mathcal{Y}^{'}_{i}\Vert^{2}_{2} \\&+\frac{1}{2}\Vert \mathcal{Y}^{fuse}_{i} - \mathcal{Y}^{self}_{i}\Vert^{2}_{2}& {\raisebox{.5pt}{\textcircled{\raisebox{-.9pt} {2}}}}
\end{split}
\\
\end{array}
\right.
\end{equation}
\end{proposition}

\begin{proof}
Given the half-quadratic problem:
\begin{equation}\label{equation:HQ}
\left\{\begin{array}{ll}
\begin{split}
\  \mathcal{W}^{'}_{i}= &\arg min_{\mathcal{W}^{'}_{i}} \frac{\beta}{2}\Vert \mathcal{W}_{i-1} - \mathcal{W}^{'}_{i}\Vert^{2}_{2} 
\\&+ \eta \; \mathcal{P}_{x}(\mathcal{W}^{'}_{i}) & {\raisebox{.5pt}{\textcircled{\raisebox{-.9pt} {1}}}}
\\
\  \mathcal{W}_{i} = &\arg min_{\mathcal{W}_{i}} \frac{\beta}{2} \Vert \mathcal{W}_{i} - \mathcal{W}^{'}_{i}\Vert^{2}_{2} \\&+\frac{1}{2}\Vert \mathcal{W}_{i} \otimes Z - \mathcal{Y}_{i}\Vert^{2}_{2}& {\raisebox{.5pt}{\textcircled{\raisebox{-.9pt} {2}}}}
\end{split}
\\
\end{array}
\right.
\end{equation}
We can transfer the sub-problem (\ref{equation:HQ})-\(\raisebox{.5pt}{\textcircled{\raisebox{-.9pt} {1}}}\) to:
\begin{equation}\label{equ:app-re1}
\begin{split}
     \mathcal{W}^{'}_{i} \otimes Z= &\arg min_{\mathcal{W}^{'}_{i}\otimes Z} \frac{\beta}{2}\Vert \mathcal{W}_{i-1} \otimes Z - \mathcal{W}^{'}_{i} \otimes Z\Vert^{2}_{2} 
\\&+ \eta \; \mathcal{P}_{x}(\mathcal{W}^{'}_{i} \otimes Z)
\end{split}
\end{equation}
We define $\mathcal{Y}^{'}_{i} = softmax (\sum_{m = 1}^{M} w^{'m}_{i} \cdot z^{m} + p_{u}) = \mathcal{W}^{'}_{i} \otimes Z$, and since
\begin{equation}
\begin{split}
    y^{fuse} = \tilde{z} \odot z &= softmax (\sum_{m = 1}^{M} {\rm log} \; (z^{m}) \cdot \mathcal{T}  (\tilde{z}^{m}) + p_{u})\\
    &= \mathcal{W} \otimes Z,
\end{split}
\end{equation}
we can rewrite (\ref{equ:app-re1})-\(\raisebox{.5pt}{\textcircled{\raisebox{-.9pt} {1}}}\) as:
\begin{equation}\label{equ:app-res1}
\begin{split}
\  \mathcal{Y}^{'}_{i}= &\arg min_{\mathcal{Y}^{'}_{i}} \frac{\beta}{2}\Vert \mathcal{Y}^{fuse}_{i-1} - \mathcal{Y}^{'}_{i}\Vert^{2}_{2} 
\\&+ \eta \; \mathcal{P}_{x}(\mathcal{Y}^{'}_{i}).
\end{split}
\end{equation}

Likewise, we can transfer the problem of (\ref{equation:HQ})-\(\raisebox{.5pt}{\textcircled{\raisebox{-.9pt} {2}}}\) to:
\begin{equation}\label{equ:app-re2-1}
\begin{split}
     \mathcal{W}_{i} \otimes Z= &\arg min_{\mathcal{W}_{i}} \frac{\beta}{2} \Vert \mathcal{W}_{i} \otimes Z - \mathcal{W}^{'}_{i} \otimes Z\Vert^{2}_{2} + \\
     &\frac{1}{2}\Vert \mathcal{W}_{i} \otimes Z - \mathcal{Y}_{i}\Vert^{2}_{2}\\
\end{split}
\end{equation}
Thus Eqn. \ref{equ:app-re2-1} can be represented as:
\begin{equation}\label{equ:app-re2-2}
\begin{split}
     \mathcal{Y}^{fuse}_{i}= &\arg min_{\mathcal{Y}^{fuse}_{i}} \frac{\beta}{2} \Vert \mathcal{Y}^{self}_{i} - \mathcal{Y}^{'}_{i}\Vert^{2}_{2} + \\
     &+\frac{1}{2}\Vert \mathcal{Y}^{fuse}_{i} - \mathcal{Y}_{i}\Vert^{2}_{2}\\
\end{split}
\end{equation}
Since $\mathcal{Y}$ is defined as the fusion of $\mathcal{W}$ and $\mathcal{Z}$ when it gets the closest to $\Vert p_{u} - p_{x}\Vert^{2}_{2}$, we represent $\mathcal{Y}$ as an approximate self-fusion $\mathcal{W} \otimes e^{\mathcal{W}}$ (according to Proposition \ref{pro:1}, $e^{\mathcal{W}} = \tilde{Z}$) with a compensation of this approximation which related to the image prior. Since $\mathcal{Y}^{'}$ is optimized under the constraint of the image prior, we take $\mathcal{Y}^{'}$ to construct the compensation and reorganize Eqn. \ref{equ:app-re2-2} as:
\begin{equation}\label{equ:app-re2-3}
\begin{split}
     \mathcal{Y}^{fuse}_{i} = 
     & \arg min_{\mathcal{Y}^{fuse}_{i}} \frac{\beta}{2} \Vert \mathcal{Y}^{fuse}_{i} - \mathcal{Y}^{'}_{i}\Vert^{2}_{2} + \\
     &+\frac{1}{2} (\Vert \mathcal{Y}^{fuse}_{i} - \mathcal{Y}^{self}_{i}\Vert^{2}_{2} + \Vert \mathcal{Y}^{self}_{i} - \mathcal{Y}^{'}_{i}\Vert^{2}_{2})\\
     = & \arg min_{\mathcal{Y}^{fuse}_{i}} \frac{\beta + 1}{2} \Vert \mathcal{Y}^{fuse}_{i} - \mathcal{Y}^{'}_{i}\Vert^{2}_{2} \\&+\frac{1}{2}\Vert \mathcal{Y}^{fuse}_{i} - \mathcal{Y}^{self}_{i}\Vert^{2}_{2}
\end{split}
\end{equation}

Thus, we have \ref{equation:HQ} is equivalent to:
\begin{equation}
\left\{\begin{array}{ll}
\begin{split}
\  \mathcal{Y}^{'}_{i}= &\arg min_{\mathcal{Y}^{'}_{i}} \frac{\beta}{2}\Vert \mathcal{Y}^{fuse}_{i-1} - \mathcal{Y}^{'}_{i}\Vert^{2}_{2} 
\\&+ \eta \; \mathcal{P}_{x}(\mathcal{Y}^{'}_{i}) & {\raisebox{.5pt}{\textcircled{\raisebox{-.9pt} {1}}}}
\\
\  \mathcal{Y}^{fuse}_{i} = &\arg min_{\mathcal{Y}^{fuse}_{i}} \frac{\beta + 1}{2} \Vert \mathcal{Y}^{fuse}_{i} - \mathcal{Y}^{'}_{i}\Vert^{2}_{2} \\&+\frac{1}{2}\Vert \mathcal{Y}^{fuse}_{i} - \mathcal{Y}^{self}_{i}\Vert^{2}_{2}& {\raisebox{.5pt}{\textcircled{\raisebox{-.9pt} {2}}}}
\end{split}
\\
\end{array}
\right.
\end{equation}
\end{proof}

\end{document}